\documentclass[acmsmall]{acmart}

\usepackage[utf8]{inputenc} 
\usepackage[T1]{fontenc}    
\usepackage{hyperref}       
\usepackage{url}            
\usepackage{booktabs}       
\usepackage{amsfonts}       
\usepackage{amstext}
\usepackage{amsmath,amsthm}
\usepackage{nicefrac}       
\usepackage{microtype}      
\usepackage{xcolor}         
\usepackage{natbib}
\setcitestyle{numbers,square}
\usepackage[linesnumbered,ruled]{algorithm2e}   
\usepackage{multirow}
\usepackage{graphicx}
\usepackage{comment}
\usepackage{float}
\usepackage{tablefootnote}
\usepackage{hyperref}

 \newtheorem{theorem}{\textbf{Theorem}}
 \newtheorem{proposition}{\textbf{Proposition}} \newtheorem{lemma}{\textbf{Lemma}}
 \newtheorem{definition}{\textbf{Definition}}
\newtheorem{assumption}{\textbf{Assumption}}
\newtheorem{remark}{Remark}
\newtheorem{claim}{\textbf{Claim}}
\usepackage{hyperref} 
\usepackage{lineno}

\AtBeginDocument{%
  \providecommand\BibTeX{{%
    \normalfont B\kern-0.5em{\scshape i\kern-0.25em b}\kern-0.8em\TeX}}}

\begin{document}

\setcopyright{acmlicensed}
\acmJournal{PACMMOD}
\acmYear{2024} \acmVolume{2} \acmNumber{3 (SIGMOD)} \acmArticle{178} \acmMonth{6}\acmDOI{10.1145/3654981}

\title{Towards Metric DBSCAN: Exact, Approximate, and Streaming Algorithms}


\author{Guanlin Mo}
\authornote{Both authors contributed equally to this research.}
\email{moguanlin@mail.ustc.edu.cn}
\author{Shihong Song}
\authornotemark[1]
\email{shihongsong@mail.ustc.edu.cn}
\affiliation{%
  \institution{University of Science and Technology of China}
  \streetaddress{No.96, JinZhai Road Baohe District}
  \city{Hefei}
  \state{Anhui}
  \country{China}
  \postcode{230026}
}

\author{Hu Ding}
\authornote{Corresponding author.}
\affiliation{%
  \institution{University of Science and Technology of China}
  \streetaddress{No.96, JinZhai Road Baohe District}
  \city{Hefei}
  \state{Anhui}
  \country{China}}
  \postcode{230026}
\email{huding@ustc.edu.cn}

\begin{abstract}
    DBSCAN is a popular density-based clustering algorithm that has  many different applications in practice. However, 
    the running time of DBSCAN in high-dimensional space or general metric space ({\em e.g.,} clustering a set of texts by using edit distance) can be as large as quadratic in the input size. Moreover, most of existing accelerating techniques for DBSCAN are only available for low-dimensional Euclidean space. In this paper, we study the DBSCAN problem under the assumption that the inliers (the core points and border points) have a low intrinsic dimension (which is a realistic assumption for many high-dimensional applications), 
    where the outliers can locate anywhere in the space without any assumption.
    First, we propose a $k$-center clustering based algorithm that can reduce the time-consuming labeling and merging tasks of DBSCAN to be linear. Further, we propose a linear time approximate DBSCAN algorithm, where the key idea is building a novel small-size summary for the core points. Also, our algorithm can be efficiently implemented for streaming data and the required memory is independent of the input size.  Finally, we conduct our experiments and compare our algorithms with several popular DBSCAN algorithms. The experimental results suggest that our proposed approach can significantly reduce the computational complexity in practice. Our source code can be found at \href{https://github.com/MoGuanlin/Towards-Metric-DBSCAN}{https://github.com/MoGuanlin/Towards-Metric-DBSCAN}.
\end{abstract}





\maketitle
\section{INTRODUCTION}
\label{sec-intro}
{\em Density-based clustering} is an important clustering model for many real-world applications~\citep{kriegel2011density}. Roughly speaking, the goal of density-based clustering is to identify the dense subsets from a given data set in some metric space. 
Unlike other clustering algorithms ({\em e.g.,} the center-based clustering algorithms like $k$-means clustering~\citep{tan2016introduction}), 
the density-based clustering algorithms are particularly useful to extract clusters of arbitrary shapes with outliers. So they also can be applied to   solve outlier recognition problems~\citep{chandola2009anomaly}. 
In the past decades, a number of  density-based clustering algorithms have been proposed~\citep{ester1996density,ankerst1999optics,mcinnes2017hdbscan,bhattacharjee2021survey}.
\textbf{DBSCAN} (Density Based Spatial Clustering
of Applications with Noise), which was proposed by 
\citet{ester1996density}, is one of the most popular density-based clustering algorithms and has been extensively applied to various practical areas~\citep{gan2015dbscan,schubert2017dbscan,he2014mr}. Due to its substantial contributions in  practice, the original DBSCAN paper was awarded {\em the test of time award} in ACM KDD 2014.


Though the DBSCAN method often achieves promising clustering performance in practice, the high computational complexity can be prohibitive for   applications in large-scale data. 
 Suppose the input data size is $n$. The worst-case time complexity of DBSCAN is $\Theta(n^2)$. 
 Obviously this quadratic complexity is far from satisfying for large $n$ in practice, and thus it motivated several faster implementations of DBSCAN in the past years. 
\citet{gunawan2013faster} proposed a grid-based DBSCAN algorithm,
which can reduce the time complexity to $O(n\log n)$ in $2D$. 
For general low-dimensional Euclidean space $\mathbb{R}^d$ (assume $d$ is small),  \citet{chen2005geometric} and  \citet{gan2015dbscan} independently proposed similar sub-quadratic complexity algorithms by using some computational geometry techniques. For example, the time complexity of \citep{gan2015dbscan} is $O(n^{2-\frac{2}{\lceil d/2\rceil+1}+\delta})$ with $\delta$ being any arbitrarily small non-negative constant. 
However, as $d$ increases, 
the time complexities of \citep{chen2005geometric,gan2015dbscan} approach to $O(n^2)$.
Furthermore, \citet{gan2015dbscan} proved that the DBSCAN problem is at least as hard as the {\em unit-spherical emptiness checking (USEC)} problem \citep{agarwal1990euclidean},
where it is widely believed that the time complexity for solving the USEC problem is  $\Omega(n^{\frac{4}{3}})$ for any $d\geq 3$~\citep{erickson1995relative,erickson1995new}.
Therefore, most recent works { mainly} focus on the low-dimensional DBSCAN problem~\citep{sarma2019mudbscan,song2018rp,wang2020theoretically,de2019faster}.

Due to the  hardness of the exact DBSCAN problem, another line of research on approximate DBSCAN attracts more attentions recently. 
For example, the aforementioned works 
\citep{chen2005geometric}  and \citep{gan2015dbscan} also 
proposed approximate DBSCAN algorithms with theoretical guarantees.  
The connectivity between points within each cluster, which is the most time-consuming step in exact DBSCAN,  is relaxed in their algorithms.
Their algorithms both achieve the $O(n/\rho^{d-1})$ time complexity, where $\rho>0$ is the given parameter that measures the approximation quality (we provide the formal definition in Section \ref{dbscan_def}).
Other variants of approximate DBSCAN methods include sampling based DBSCAN++~\citep{jang2019dbscan++},
LSH based DBSCAN~\citep{wu2007linear}, and KNN based NG-DBSCAN~\citep{lulli2016ng}. 

However, the research on DBSCAN for abstract metric space and high-dimensional space is still quite limited, to the best of our knowledge. 
In this big data era, we often confront with high-dimensional data~\citep{fan2014challenges} and even non-Euclidean data~\citep{bronstein2017geometric}. 
For example, the input data can be texts, images, or biological sequences, 
which cannot be embedded into a low-dimensional Euclidean space; 
and moreover, the distance between different data items can be more complicated than Euclidean distance ({\em e.g.,} edit distance~\citep{navarro2001guided}).  
We should emphasize that one can always run the original DBSCAN algorithm~\citep{ester1996density} for solving DBSCAN in abstract metric space or high-dimensional space, but the previously developed accelerating techniques cannot be directly applied to reduce the complexity for these cases. 
\citet{ding2021metric} proposed a fast exact DBSCAN algorithm for abstract metric space by  using the $k$-center clustering heuristic~\citep{ding2019greedy}; but their result lacks strict theoretical analysis. \citet{lulli2016ng} and \citet{yang2019dbscan} considered distributed DBSCAN for metric space, but their methods mainly focus on the communication cost and load balance for distributed system rather than the time complexity ({\em e.g.}, the approach of \cite{yang2019dbscan} just directly uses the orginal DBSCAN algorithm in each local machine). 

\textbf{Our main results and the key ideas.} 
In this paper, we systematically study the DBSCAN problem in abstract metric space (we call it ``\textbf{metric DBSCAN}'' problem for short). Our results can be also extended to the case in high-dimensional Euclidean space, if the intrinsic dimension of the input data 
(except for the outliers) 
is low, which is a reasonable assumption for a   range of applications in real world ({\em e.g.,} image datasets~\citep{roweis2000nonlinear}).
We use the ``doubling dimension'' $D$ to measure the intrinsic dimension (the formal definition is shown in Section~\ref{dbscan_def}). Also note that we do not require the value of $D$ to be explicitly given in our algorithms. 
We do not add any constraint to the outliers, because in practice they can locate anywhere in the space. {For example, in the area of AI security, an  attacker may add some carefully crafted outliers to mislead the machine learning model, where some recent works showed that the outliers often have significantly higher intrinsic dimension than   normal data~\cite{houle2018characterizing,weerasinghe2021defending}. }
 We consider both the exact and approximate metric DBSCAN problems.  Our main contributions are twofold {\color{black}(the results can be easily extended to the case that both inliers and outliers have low doubling dimensions).} 

\textbf{(1)} The exact DBSCAN problem involves two major  tasks: identify the core points and merge the core points to form the clusters. 
Both of these two tasks take $O(n^2)$ time in the worst case if using the original DBSCAN algorithm.
First, we show that the first task can be completed with the runtime linear in $n$. 
The key idea is inspired by a novel insight that relates the task to a {\em radius-guided $k$-center clustering} method in metric space, which can reduce the search range for identifying the core points. After labeling the core points, we then build a set of cover trees~\cite{beygelzimer2006cover} for the local regions, and solve the second task for merging the core points with the time linear in $n$ as well.   

\textbf{(2)} We also consider the  $\rho$-approximate DBSCAN problem proposed in~\cite{gan2015dbscan}. 
Our idea is to construct a ``summary'' for the set of core points. 
The summary should be much smaller than the core points set (which can be as large as $n$); 
also the summary should be able to approximately represent the core points, 
so that we can efficiently recover the $\rho$-approximate DBSCAN clusters without dealing with all the core points. The approximate algorithm also has the time linear in $n$; but comparing with our proposed exact DBSCAN algorithm, it  compresses the data into a small-size summary, so the method 
can be easily extended for developing streaming algorithm 
where the required memory is independent of $n$. To the best of our knowledge, the current research for streaming DBSCAN (exact or approximate DBSCAN with theoretical quality guarantee) are still quite limited, though several other types of density-based streaming clustering algorithms were proposed before~\citep{hahsler2016clustering,chen2007density,carnein2018evoStream,fichtenberger2013bico}.
 
We also take a comprehensive set of experiments to evaluate the performance of our proposed algorithms. Our algorithms can achieve significant speedups over several existing DBSCAN algorithms. Another advantage of  our method is that it is   friendly for parameter tuning. For example, one usually needs to try different  values for the parameters ({\em e.g,} the ball radius and density threshold) to achieve a satisfied DBSCAN solution; the radius-guided $k$-center clustering, which is a key part in our algorithms,  do not need to be performed repeatedly during the parameter tuning, and therefore our overall running time can be  further reduced   in practice.

\subsection{Preliminaries}

{
\color{black}

}

\label{dbscan_def}
\textbf{Notations.} We denote the metric space  by $(X,\mathtt{dis})$,  where $X$ is the set of $n$ input data items and $\mathtt{dis}(\cdot,\cdot)$ is the distance function on $X$. 
For any $p \in X$ and $Y \subseteq X$, we define $\mathtt{dis}(p, Y) = \min_{y \in Y} \mathtt{dis}(p,y)$.
Let $\mathbb{B}(p,r)$ be the ball centered at a point $p\in X$ with radius $r \ge 0$ in the space. For any  $Y\subseteq X$, we use $|Y|$ to denote the number of data items in $Y$. 
We assume that it takes $t_{\mathtt{dis}}$ time to compute the distance $\mathtt{dis}(p_1, p_2)$ between any two points $p_1, p_2\in X$. For example, $t_{\mathtt{dis}}=O(d)$ for $d$-dimensional Euclidean space.  

\subsubsection{Overview for DBSCAN}
\label{sec-overviewdbscan}
For the sake of completeness, we briefly overview the DBSCAN algorithm~\citep{ester1996density} and the $\rho$- approximation \citep{gan2015dbscan}. 
DBSCAN has two input parameters: $\epsilon \in \mathbb{R}^+$ and $MinPts \in \mathbb{Z}^+$.
All the input data points of $X$ are divided  into three categories (see Figure~\ref{dbscan-fig} for an illustration): for any $p\in X$, 
{\color{black}
\begin{enumerate}
      \item $p$ is a \textbf{core point} if it has at least $MinPts$ neighbors within the distance $\epsilon$, {\em i.e.,} $|\mathbb{B}(p,\epsilon) \cap X| \ge MinPts$; 
  \item $p$ is a \textbf{border point} if $p$ is not a core point, but {\color{black} there exists at least one core point $q$ satisfying $p\in \mathbb{B}(q,\epsilon)$;}
  \item  $p$ is an \textbf{outlier}, if it does not satisfy neither (1) nor (2).
\end{enumerate}
}
\vspace{-5pt}
\begin{figure}[H]
    \centering
    \includegraphics[width=0.3\textwidth]{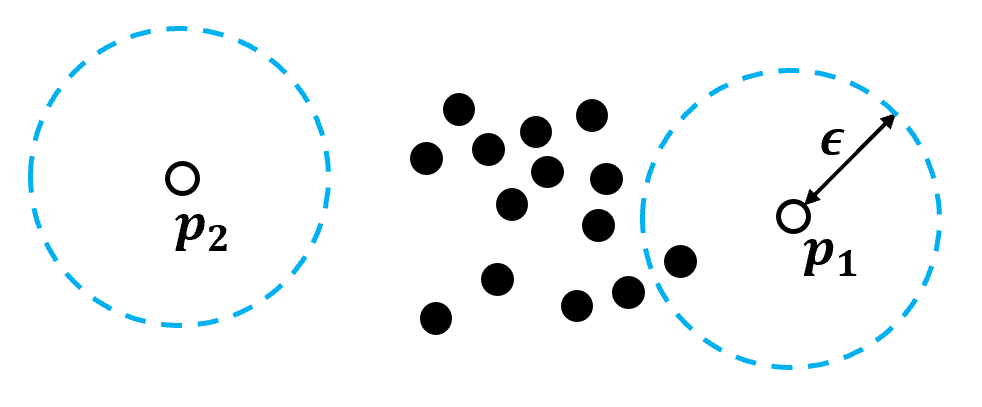}
    \vspace{-5pt}
    \caption{{An example for DBSCAN with $MinPts=4$: the solid points are core points, the point $p_1$ is a border point, and the point  $p_2$ is an outlier.}}
    \label{dbscan-fig}
\end{figure}
\vspace{-5pt}


{\color{black} We also need to consider the ``reachability'' to define DBSCAN clusters. Let $q$ be a core point. A point $p \in X$ is \textbf{density-reachable} from $q$ if we can find a sequence of points $p_1,p_2,\cdots,p_{t} \in X$ satisfying: all of them  (with the possible exception of $p_t$) are core points, and
\begin{itemize}
        \item $p_{i+1} \in \mathbb{B}(p_i,\epsilon)$ for each $i\in [1,t - 1]$;
    \item $p_1 = q$ and $p_t = p$.
\end{itemize}
}

\begin{definition}[DBSCAN~\citep{ester1996density}]
\label{def-dbscan}
  We say a non-empty set of points $C \subseteq X$ is a \textbf{cluster}, if $C$ satisfies the following two conditions:
\textbf{(1) maximality:} $\forall p, q$, if $q \in C$ and $p$ is density-reachable from $q$, then $p \in C$;
    \textbf{(2) connectivity:} for any pair of points $p,q\in C$, there is a point $o \in C$ such that both $p$ and $p$ are  
    density-reachable from $o$.
  The \textbf{DBSCAN problem} is to find the  set of DBSCAN clusters of $X$\footnote{In this definition, a border point may be assigned to multiple clusters.}
  \end{definition}

Let $\rho>0$. We say a point $p \in X$ is \textbf{$\rho$-approximate density-reachable} from a core point $q$ if   the third condition of density-reachable ``$p_{i+1} \in \mathbb{B}(p_i,\epsilon)$'' is relaxed to be  ``$p_{i+1} \in \mathbb{B}(p_i,(1+\rho)\epsilon)$'' for each $i=1,2,\cdots,t-1$.

\begin{definition}[$\rho$-approximate DBSCAN~\citep{gan2015dbscan}]
\label{def-rhoapprox}
By using ``$\rho$- approximate density-reachable'' instead of ``density-reachable'' in the  \textbf{connectivity}  condition, one can define  \textbf{$\rho$-approximate cluster} via the similar manner of Definition~\ref{def-dbscan}. 
The \textbf{$\rho$-approximate DBSCAN problem} is to find a set of $\rho$-approximate clusters such that 
  every core point of $X$ belongs to exactly one cluster.
\end{definition}

It is easy to see that $\rho$-approximate DBSCAN is equivalent to the exact DBSCAN problem if $\rho=0$. Also, Gan and Tao~\citep{gan2015dbscan} proved a ``Sandwich Theorem'' for it. Namely, the $\rho$-approximate DBSCAN solution $S$ with $(\epsilon, MinPts)$ is sandwiched by the exact DBSCAN solutions $S_1$ and $S_2$ with $(\epsilon, MinPts)$ and $((1+\rho)\epsilon, MinPts)$ respectively: for any two points $p, q\in X$, if they fall into the same cluster of $S_1$, they must also fall into the same cluster of $S$; if they fall into the same cluster of $S$, they must also fall into the same cluster of $S_2$.

%
%
\begin{remark}
\label{rem-minpts}
The parameter $MinPts$ is usually set to be a constant (say, less than $ 10$) as recommended in previous articles like~\citep{ester1996density, gunawan2013faster, de2019faster}. So we assume that $MinPts$ is a small number throughout this paper. 
\end{remark}

\subsubsection{Doubling Dimension}
\label{sec-dd}
We  take ``doubling dimension'', which is a natural and robust measure that has been widely used in machine learning and data analysis~\citep{robinson2010dimensions,gupta2003bounded}, to indicate the intrinsic dimension of data. Intuitively, the doubling dimension measures how fast the data volume is growing in the space. 
{\color{black}
\begin{definition}[Doubling dimension \cite{gupta2003bounded}]
  \label{Doubling Dimension}
  Given a metric space $(X, \mathtt{dis})$, let $\Lambda$ be the smallest positive integer such that for any $p\in X$ and any $r \ge 0$,
  $X \cap \mathbb{B}(p,2r)$ can always be covered by the union of at most $\Lambda$ balls with radius $r$. Then the doubling dimension of $(X, \mathtt{dis})$ is $D=\lceil\log_2 \Lambda\rceil$. 
\end{definition}}
We also have the following proposition for doubling dimension that is used in our following analysis. 

{
\color{black}
\begin{proposition}[ \citep{talwar2004bypassing,krauthgamer2004navigating} ]
  \label{Ap_lemma}
  Suppose the doubling dimension of a metric space $(X, \mathtt{dis})$ is $D$. For any point set $Y \subseteq X$, we have $|Y|\le 2^{D\lceil \log \alpha \rceil}$, where $\alpha$ is the aspect ratio of $Y$, \em{i.e.,} 
  $\alpha = \frac{\max_{y,y'\in Y}\mathtt{dis}(y,y')}{\min_{y,y'\in Y}\mathtt{dis}(y,y')}$. 
\end{proposition}
}

 We let $X_{\mathtt{out}}$ be the set of DBSCAN outliers and $X_{\mathtt{in}}=X\setminus X_{\mathtt{out}}$.  
For convenience, we let $z = |X_{\mathtt{out}}|$. 
We have Assumption~\ref{ass-doubling2} following the previous articles on algorithms design in doubling metric~\citep{DBLP:conf/focs/HuangJLW18,ceccarello2019solving,ding2021metric}. 

\begin{assumption}
\label{ass-doubling2}
We assume that the doubling dimension $D$ of the sub-metric space $(X_{\mathtt{in}}, \mathtt{dis})$ is constant. 
\end{assumption}
 For the outliers $X_{\mathtt{out}}$, we do not make any assumption for their doubling dimension. In other words, the outliers can scatter arbitrarily without any constraint in the space ({\em e.g.,} the outliers may be manipulated by an attacker~\cite{houle2018characterizing,weerasinghe2021defending}).

\subsubsection{Cover Tree}
\label{sec-ct}
We also introduce a spatial query structure ``cover tree'' which was proposed by \citet{beygelzimer2006cover} for data with low intrinsic dimension. To state the definition, we need to introduce  ``$r$-net'' first~\cite{clarkson2006building}.

 \begin{definition}[$r$-net]
 Given a number $r>0$ and a set $P$ in some metric space, 
    a subset $Q \subseteq P$ is an $r$-net of $P$ if it satisfies two requirements: 
    \begin{enumerate}
        \item {$r$-packing: $\mathtt{dis}(q, Q\setminus \{q\})\ge r$ for all $q\in Q$.}
        \item  $r$-covering: $\mathtt{dis}(p,Q)\le r$ for all $p\in P$.
    \end{enumerate}
    \label{r-net}
\end{definition}
\vspace{-5pt}

A \textbf{Cover Tree} is a hierarchical tree structure $T$ for storing a given data set $P$, where each node corresponds to a point of $P$. 
The set of nodes at the $i$-th level is denoted by $T_i$ for   $l_{\mathtt{bottom}}\leq i\leq l_{\mathtt{top}}$, where $l_{\mathtt{bottom}}$ and $l_{\mathtt{top}}$ indicate the lowest and  highest   levels of $T$, respectively; $T_{l_{\mathtt{top}}}$ is the root node, and $T_i \subset T_{i-1}$ for each $i>l_{\mathtt{bottom}}$. Also, we require that each $T_i$ is a $2^i$-net of $T_{i-1}$.

 For convenience, let $\Delta_P$ and $\delta_P$ denote the maximum and minimum pairwise distances of $P$, respectively.
From the above definition of the cover tree, we can derive that $l_{\mathtt{top}}\leq \lceil \log \Delta_P \rceil$ and $l_{\mathtt{bottom}}\geq \lfloor \log \delta_P \rfloor$ (so $l_{\mathtt{bottom}}$ can be negative if $\delta_P <1$).


  
{
\begin{claim}[complexities of cover tree]
\label{cover_insert}
    Let  the aspect ratio of $P$ be $\Phi_P=\Delta_P/\delta_P$. Suppose the dataset $P$ has a doubling dimension $D$. The construction time of the cover tree of $P$ is {$O(2^{O(D)}|P|\log \Phi_P \cdot t_{\mathtt{dis}})$.} 
    For any query point $p$,   finding its nearest neighbor  in $P$ takes $O(2^{O(D)}\log \Phi_P \cdot t_{\mathtt{dis}})$ time.
    \end{claim}
}


    

\begin{remark}
In the original article~\cite{beygelzimer2006cover}, the authors assumed that the input data has a constant ``expansion rate'', which is a similar but more stringent  intrinsic dimension measure than  the doubling dimension~\cite[Proposition 1.2]{gupta2003bounded}. For example, a data set has a constant expansion rate must also have a constant doubling dimension, but not vice versa; in other words, doubling dimension is a more general measure than the expansion rate. The complexities in terms of the dobuling dimension of Claim~\ref{cover_insert} can be obtained by using the Nevigating net idea \cite{elkin2023new,krauthgamer2004navigating}.

\end{remark}

In recent years, several new variants of the cover tree, such as the simplified cover tree \citep{izbicki2015faster} and the compressed cover tree \citep{elkin2023new}, were also proposed. In this paper, we only use the vanilla cover tree for simplicity in our analysis, and one can also replace it by those variants in practice.

\textbf{The rest of this paper is organized as follows.}
In Section~\ref{gonzalez}, we propose the radius-guided $k$-center clustering algorithm that is the key technique for our algorithms. In Section~\ref{exact_dbscan}, we show our algorithm for exact metric DBSCAN problem. In Section~\ref{approx_dbscan}, we show our algorithm and its streaming version for $\rho$-approximate DBSCAN. 
Finally, we illustrate our experimental results in Section~\ref{sec-exp}. 
\section{Radius-guided Gonzalez's algorithm}
\label{gonzalez}
In this section, we revisit the Gonzalez's algorithm \citep{gonzalez1985clustering} which is a popular method for solving the $k$-center clustering problem. 
The goal of $k$-center clustering is 
 to find $k$ balls to cover the whole data with minimizing the maximum radius. 
The Gonzalez's algorithm is a standard greedy procedure. It iteratively selects $k$ points, where each point has the largest distance to the set of the points that were selected in the previous rounds. The algorithm yields a $2$-approximate solution, that is, the obtained radius is no larger than two times the optimal one. It was also shown that designing a $(2-\gamma)$-approximate algorithm for any $\gamma > 0$ is NP hard~\citep{hochbaum1986unified}.


In this section, we propose a radius-guided version of the Gonzalez's algorithm (Algorithm~\ref{exp_kcenter}).
The key difference between our algorithm and the vanilla Gonzalez's algorithm is that 
the former one needs an upper bound $\bar{r}$ of the radius instead of  the number of centers $k$.



\vspace{-5pt}
\begin{algorithm}
  \SetAlgoLined
  \caption{\sc{Radius-guided Gonzalez}}
  \label{exp_kcenter}
  \KwIn{The dataset $X$ and an upper bound $\bar{r}>0$}

  Take an arbitrary point $p_0\in X$ and set $E=\{p_0\}$
  \tcp*[r]{$E$ stores the center points.}
  Set $d_{\mathtt{max}} = \max_{p \in X} \mathtt{dis}(p, E)$.

  For any $p \in X$,  define its closest center $c_p = \arg \min_{e\in E} \mathtt{dis}(p,e)$ (if there is a tie, we arbitrarily pick one as $c_p$). \tcp*[r]{We call $c_p$ as the center of $p$.}   
  
  
  \While{$d_{\mathtt{max}} > \bar{r}$}{
    Let $q = \arg \max_{p \in X} \mathtt{dis}(p, E)$, and 
    add $q$ to $E$.

    Update $d_{\mathtt{max}}$, $c_p$ and $\mathtt{dis}(p, E)$ for each $p\in X$, 
    and the set $\mathcal{C}_e=\{p\mid p\in X \text{ \& } c_p=e \}$ for each $e\in E$. \tcp*[r]{We call the set $\mathcal{C}_e$ as the cover set of $e$, and it is used in our following DBSCAN algorithms.}  }
  \Return{E}
\end{algorithm}
\vspace{-5pt}
%
%
%
  Below we analyze the complexity of Algorithm~\ref{exp_kcenter} under  Assumption~\ref{ass-doubling2} introduced in Section~\ref{sec-dd}. 
  \begin{lemma}
  \label{kcenter_iteration}
  {Suppose $\Delta$ is the maximum pairwise distance of $X_{\mathtt{in}}$} and 
  $z=|X_{\mathtt{out}}|$.
  Algorithm \ref{exp_kcenter} runs at most $O\big{(}(\frac{\Delta}{\bar{r}})^D\big{)} + z$ iterations, and each iteration takes $O(nt_{\mathtt{dis}})$ time.
\end{lemma}

\begin{proof}
The time complexity for each iteration is easy to obtain. When a new point $q$ is added to $E$, we can update $\mathtt{dis}(p, E)=\min\{\mathtt{dis}(p, E\setminus\{q\}), \mathtt{dis}(p, q)\}$ in $O(t_{\mathtt{dis}})$ time for each $p\in X$; the value $d_{\mathtt{max}}$ then can be updated in $O(n)$ time. Similarly, we can update    all the centers $c_p$s  and cover sets $\mathcal{C}_e$s  in $O(n) $ time.  So the overall time complexity for each iteration is $O(nt_{\mathtt{dis}})$.

Now we focus on the maximum number of iterations.
We first consider the inliers $X_{\mathtt{in}}$.
According to Algorithm~\ref{exp_kcenter}, for any $e_1, e_2 \in E$, we claim that $\mathtt{dis}(e_1, e_2) \ge \bar{r}$; otherwise, Algorithm~\ref{exp_kcenter} should terminate earlier, and $e_1$ or $e_2$ cannot be added to $E$. Thus $\min_{e_1, e_2 \in E}$ $ \mathtt{dis}(e_1,e_2) \ge \bar{r}$.
The upper bound of the distance between two points in $ E \cap X_{\mathtt{in}} $ is $\Delta$. 
So, according to Proposition~\ref{Ap_lemma}, $ |E \cap X_{\mathtt{in}}| \le 2^{D\lceil \log \frac{\Delta}{\bar{r}} \rceil}$. Also,  we have $ |E \cap X_{\mathtt{out}}| \le z$, where $z = |X_{\mathtt{out}}|$. 
Therefore, $ |E|=|E \cap X_{\mathtt{in}}|+|E \cap X_{\mathtt{out}}| \le 2^{D\lceil \log \frac{\Delta}{\bar{r}} \rceil} + z = O\big{(}(\frac{\Delta}{\bar{r}})^D\big{)} + z$, and Algorithm~\ref{exp_kcenter} runs for $|E|$ iterations.
\end{proof}  
{
\begin{remark}
We would like to emphasize that the number of iterations of Algorithm~\ref{exp_kcenter} usually is much lower than the theoretical bound in Lemma~\ref{kcenter_iteration}, especially when  applying it to the DBSCAN problem. 
For example, a DBSCAN cluster can have any shape {({\em e.g.,} a ``banana'' shape as shown in Figure \ref{perform})} rather than a ball shape, {\em i.e.,} does not fill an entire ball; so the size upper bound ``$2^{D\lceil \log \alpha \rceil}$'' in Proposition~\ref{Ap_lemma} can be much higher than the real size, which also implies a lower number of iterations in Algorithm~\ref{exp_kcenter}. 
 \end{remark}
}

 We also introduce two other lemmas relevant to Algorithm~\ref{exp_kcenter}, which are used in our following analysis.  For any $p \in X$, its \textbf{neighbor ball center set} is:
  \begin{eqnarray}
 \mathcal{A}_p = \{e | e \in E, \mathtt{dis}(e,c_p) \le 2 \bar{r} + \epsilon\}, \label{for-ap}
 \end{eqnarray}
  where $\epsilon$ is  the parameter of DBSCAN described in Section~\ref{dbscan_def}. Note that for each $p\in X$, the set $\mathcal{A}_p$ can be simultaneously  obtained without increasing the complexity of Algorithm~\ref{exp_kcenter}. For each $q\in E$, we just keep a set $\{e\mid e\in E, \mathtt{dis}(e, q)\leq 2\bar{r}+\epsilon\}$. Note that the center set $E$ is generated incrementally in Algorithm~\ref{exp_kcenter}. When a new point $e$ is added to $E$, we just update the set for each $q\in E$. Then the set $\mathcal{A}_p$ is also obtained for each $p\in X$ when the algorithm terminates.  {\color{black}The approach of ~\citep{ding2021metric} also defines the neighbor region for each point $p$ in a similar way as (\ref{for-ap}). But a major difference is that their approach needs  an estimated upper bound ``$\tilde{z}$'' for the number of outliers  in their definition, which could be hard to obtain; moreover, it also needs to manually set the termination condition (please see Section~\ref{remark33} for the detailed discussion).}

%
\begin{lemma}
  \label{neighbor_lemma}
  For any $p\in X$, we have $X \cap \mathbb{B}(p, \epsilon) \subseteq \cup_{e\in \mathcal{A}_p}\mathcal{C}_e$.
\end{lemma}
Lemma~\ref{neighbor_lemma} can be proved by using the triangle inequality. 
Fix a point $p\in X$, and we suppose $q\in X$ with $\mathtt{dis}(c_p, c_q) > 2\bar{r} + \epsilon$  (\emph{i.e.}, $q\in \cup_{e\notin \mathcal{A}_p} \mathcal{C}_e$).
  Then we have $\mathtt{dis}(p,q) \ge \mathtt{dis}(c_p, c_q) - \mathtt{dis}(p, c_p) - \mathtt{dis}(q, c_q) > 2\bar{r} + \epsilon-2\bar{r}=\epsilon$. Therefore, we know $q\notin X \cap \mathbb{B}(p, \epsilon)$, and consequently $X \cap \mathbb{B}(p, \epsilon) \subseteq X \setminus \cup_{e\notin \mathcal{A}_p} \mathcal{C}_e =\cup_{e\in \mathcal{A}_p}\mathcal{C}_e$.
From Lemma~\ref{neighbor_lemma}, we know that it is sufficient to only consider the set $\cup_{e\in \mathcal{A}_p}\mathcal{C}_e$  for determining that whether $p$ is a core point.
This is particularly useful to reduce the neighborhood query complexity for DBSCAN.
\begin{lemma}
  \label{Ap_num}
  For each point $p\in X$, we have $|\mathcal{A}_p| = O\big{(}(\frac{\epsilon}{\bar{r}})^D + z\big{)}$.
\end{lemma}
\begin{proof}
Note that $|\mathcal{A}_p|=|\mathcal{A}_p\cap X_{\mathtt{in}}|+|\mathcal{A}_p\cap X_{\mathtt{out}}|$. Since $z=|X_{\mathtt{out}}|$, we have $|\mathcal{A}_p\cap X_{\mathtt{out}}|\leq z$. So we only need to consider the size $|\mathcal{A}_p\cap X_{\mathtt{in}}|$.

Similar to the proof of Lemma~\ref{kcenter_iteration}, we have $\min_{e_1, e_2 \in E} \mathtt{dis}(e_1,e_2) \ge \bar{r}$. 
For any $p_1, p_2 \in \mathcal{A}_p$, we have
  $\mathtt{dis}(p_1, p_2) \le \mathtt{dis}(p_1, p) + \mathtt{dis}(p, p_2)
      \le 4 \bar{r} + 2\epsilon$, where 
  the first inequality comes from the triangle inequality and the second inequality comes from the definition of $\mathcal{A}_p$. 
  According to Proposition~\ref{Ap_lemma}, we have $|\mathcal{A}_p \cap X_{\mathtt{in}}| \le 2^{D\lceil \log \alpha \rceil}$,
  where
  \begin{equation}
    \label{alpha_forlula}
    \begin{aligned}
      \alpha &= \frac{\max_{e_1,e_2 \in \mathcal{A}_p \cap X_{in}} \mathtt{dis}(e_1,e_2)}{\min_{e_1,e_2 \in \mathcal{A}_p \cap X_{in}} \mathtt{dis}(e_1,e_2)} \\
      &\le \frac{\max_{e_1,e_2 \in \mathcal{A}_p \cap X_{in}} \mathtt{dis}(e_1,e_2)}{\min_{e_1,e_2 \in E} \mathtt{dis}(e_1,e_2)} 
      \le \frac{4 \bar{r} + 2\epsilon}{\bar{r}}.
    \end{aligned}
  \end{equation} 
  So we have 
  $|\mathcal{A}_p \cap X_{\mathtt{in}}| \le 2^{\lceil \log \frac{4 \bar{r} + 2\epsilon}{\bar{r}} \rceil D}
  = O\big{(}(\frac{\epsilon}{\bar{r}})^D\big{)}$. 
  Therefore, $|\mathcal{A}_p| =   O\big{(}(\frac{\epsilon}{\bar{r}})^D + z\big{)}$.
\end{proof}


%

\begin{remark}
\label{rem-ap}
In the proof of Lemma~\ref{Ap_num}, we directly count all the $z$ points of $X_{\mathtt{out}}$ to $\mathcal{A}_p$. Actually, this is overly conservative in practice since the outliers usually are scattered and far away from the inliers. So $|\mathcal{A}_p\cap X_{\mathtt{out}}|$ usually is  much less than $ z$,  and thus  $|\mathcal{A}_p|$ is also less than the theoretical bound $O\big{(}(\frac{\epsilon}{\bar{r}})^D + z\big{)}$.
 \end{remark}

\section{A Faster Exact Metric DBSCAN Algorithm}
\label{exact_dbscan}
From Section~\ref{sec-overviewdbscan}, we know that the DBSCAN algorithm contains three key steps: \textbf{(1)} label all the core points; \textbf{(2)} merge the core points to form the clusters; \textbf{(3)} identify the border points and outliers. In this section, we analyze these three steps separately and show that our proposed radius-guided Gonzalez's algorithm can help us to significantly reduce the overall complexity.  First, we consider the problem under Assumption~\ref{ass-doubling2}, and then show that the algorithm can be simplified if Assumption~\ref{ass-doubling2} is replaced by a stronger assumption that the whole input data (including inliers and outliers) has a low doubling dimension.

\subsection{Our Algorithm under Assumption~\ref{ass-doubling2}}
\label{exact_dbscan-1}
As the pre-processing, we run Algorithm~\ref{exp_kcenter} with $\bar{r} = \epsilon / 2$ on the given instance $X$.
We also obtain the set $\mathcal{A}_p$   for each $p\in X$. 
According to Lemma~\ref{Ap_num} and the assumption that $D$ is a constant, we have 
\begin{equation}
\label{approx_ap}
    |\mathcal{A}_p| = O(z).
\end{equation}

\textbf{Step (1): label the core points.} First, for any point $p\in X$, from the triangle inequality  we know that 
\begin{eqnarray}
\mathcal{C}_{c_p}  \subseteq\mathbb{B}(p, \epsilon ) \cap X, 
\end{eqnarray}
where $\mathcal{C}_{c_p}$ is the cover set of $c_p$ stored in Algorithm~\ref{exp_kcenter} and  the size $|\mathcal{C}_{c_p} |$  can be obtained immediately in $O(1)$ time. So if $|\mathcal{C}_{c_p} | \ge MinPts$, we know $|\mathbb{B}(p, \epsilon  ) \cap X| \ge MinPts$ as well; then we can safely label $p$ as a core point. 
For the  remaining points, we check their local regions based on $\mathcal{A}_p$. Let $p$ be an unlabeled point. According to Lemma~\ref{neighbor_lemma}, we only need to count the size
\begin{equation}
    \Big|\mathbb{B}(p, \epsilon) \cap \big(\cup_{e\in \mathcal{A}_p}\mathcal{C}_e\big) \Big|.\label{for-labelcore-1}
\end{equation}
The point $p$ is a core point if and only if the size is $\geq MinPts$. Thus a natural idea for bounding the complexity of step (1) is to prove an upper bound for the size $|\cup_{e\in \mathcal{A}_p}\mathcal{C}_e|$. Unfortunately, this size can be large in reality ({\em e.g.,} one ${C}_e$ can be very dense and  contain a large number of points). To resolve  this issue, a cute idea here is to consider the amortized query complexity over all the points of $X$; 
we  show that the time for labeling the core points is linear in $n$. 
 \begin{lemma}
 \label{lem-step1}
 The time complexity of Step (1)  is $O(  n z t_\mathtt{dis})$.
 \end{lemma}
\begin{proof}


  Define $E_1=\{e\in E\mid |\mathcal{C}_e|\ge MinPts\}$ and $E_2=E\setminus E_1$, where $E$ is the set returned by Algorithm~\ref{exp_kcenter}.
  According to the previous analysis, the time complexity of step (1) is 
  \begin{equation}
  \label{step2_com}
    \begin{aligned}
     &\sum_{e\in E_1}O(|\mathcal{C}_e|) + \sum_{e\in E_2}\sum_{{e'}\in \mathcal{A}_e}O({|\mathcal{C}_e||\mathcal{C}_{e'}|}\cdot t_\mathtt{dis}) \\
      &=O(n) +O(MinPts) \cdot \sum_{e\in E_2}\sum_{e'\in \mathcal{A}_e}O(|\mathcal{C}_{e'}|\cdot t_\mathtt{dis}),
      \end{aligned}
    \end{equation}
    where the second equality comes from the facts $\sum_{e\in E_1}O(|\mathcal{C}_e|)\leq n$ and $|\mathcal{C}_e|<MinPts$ for any $e\in E_2$.
    We focus on the  term ``$ \sum_{e\in E_2}\sum_{e'\in \mathcal{A}_e}O(|\mathcal{C}_{e'}|\cdot t_\mathtt{dis})$'' of (\ref{step2_com}). {A key observation is that $e'\in \mathcal{A}_{e}$ if and only if  $e\in \mathcal{A}_{e'}$.} 
    Using this property, 
     we can exchange the order of summation for $e$ and $e'$, {\em i.e.,}
    \begin{equation}
    \label{exchange}
        \begin{aligned}
      \sum_{e\in E_2}\sum_{e'\in \mathcal{A}_e}O(|\mathcal{C}_{e'}|\cdot t_\mathtt{dis})
      &=    \sum_{e'\in E_1 \cup E_2}\sum_{{e\in \mathcal{A}_{e'}\cap E_2}}O(|\mathcal{C}_{e'}|\cdot t_\mathtt{dis}).  
            \end{aligned}
    \end{equation}
    Note that $O(|\mathcal{C}_{e'}|\cdot t_\mathtt{dis})$ is independent with $e$, so we have
    \begin{equation}
     \label{exchange-2}
        \begin{aligned}
      (\ref{exchange}) &=  \sum_{e'\in E_1 \cup E_2}O(|\mathcal{C}_{e'}|\cdot t_\mathtt{dis}) \cdot O(|\mathcal{A}_{e'}|).\\
      &= O(n\cdot t_\mathtt{dis}) \cdot O(z) ,
      \end{aligned}
    \end{equation}
   where the last equality comes from $\sum_{e'\in E_1 \cup E_2}O(|\mathcal{C}_{e'}|) = O(n)$ and $|\mathcal{A}_{e'}|=O(z)$ in (\ref{approx_ap}). Therefore, through combining (\ref{step2_com}) and (\ref{exchange-2}), we have the total complexity of Step (1)
    \begin{equation}
        \begin{aligned}
        &= O(n) + O(MinPts) \cdot O(n\cdot t_\mathtt{dis}) \cdot O(z) \\
      &= O(MinPts \cdot n \cdot z\cdot t_\mathtt{dis}) \\
      &= O(nzt_{\mathtt{dis}}),
    \end{aligned}
  \end{equation} 
 { if assuming $MinPts$ is a constant as discussed in Remark \ref{rem-minpts}. }
\end{proof}
 
\textbf{Step (2): merge the core points to form the DBSCAN clusters.} We consider the set $E$ obtained in Algorithm~\ref{exp_kcenter}. For each $e\in E$, denote by $\tilde{\mathcal{C}}_e$   the set of core points in $ \mathcal{C}_e$ which are labeled in Step (1) (if $\tilde{\mathcal{C}}_e=\emptyset$, we just simply ignore $ \mathcal{C}_e$). 
First, all the core points of $\tilde{\mathcal{C}}_e$ should be merged into the same cluster; then through the triangle inequality, we only need to consider the core points of $\tilde{\mathcal{C}}_{e'}$ with $e' \in \mathcal{A}_{e}$. The core points of $\tilde{\mathcal{C}}_e$ and $\tilde{\mathcal{C}}_{e'}$ should be merged into the same cluster, if and only if their smallest pairwise distance $\mathtt{dis}(\tilde{\mathcal{C}}_{e}, \tilde{\mathcal{C}}_{e'})=\min_{p\in \tilde{\mathcal{C}}_{e},q\in \tilde{\mathcal{C}}_{e'}}\mathtt{dis}(p, q) \le \epsilon$. Please see Figure~\ref{merge-core-fig} for an illustration. 

\vspace{-5pt}
\begin{figure}[H]
    \centering
    \includegraphics[width=0.45\textwidth]{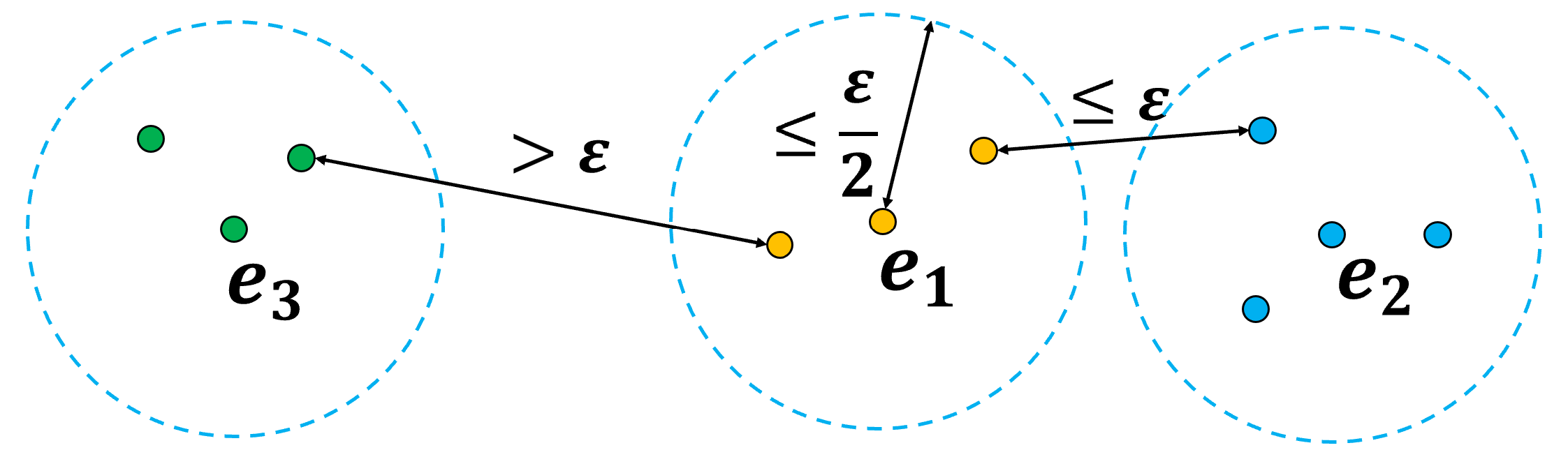}
    \vspace{-5pt}
    \caption{  {The sets $\tilde{C}_{e_1}$ (yellow points),  $\tilde{C}_{e_2}$ (blue points), and $\tilde{C}_{e_3}$ (green points) are shown in the figure. $\tilde{C}_{e_1}$ and $\tilde{C}_{e_2}$ should be merged into the same cluster, because their closest pair distance is less than $\epsilon$; on the other hand, $\tilde{C}_{e_3}$ should be merged to a different cluster}}
    \label{merge-core-fig}
\end{figure}
\vspace{-8pt}

Actually this is the  classical {\em bichromatic closest pair (BCP)} problem in computational geometry~\citep{agarwal1990euclidean}. Directly solving the BCP problem in a general metric space or a high-dimensional space can be quite challenging, {\em e.g.,} recently \citet{rubinstein2018hardness} showed that it requires at least nearly quadratic time to compute even only an approximate solution. A good news is that we have already distinguished the core points from other points in Step (1); that is, we can apply the indexing structure cover tree (Section~\ref{sec-ct}) to solve the BCP problem efficiently under Assumption~\ref{ass-doubling2}.  
Specifically, we build a cover tee for $\tilde{\mathcal{C}}_{e}$, and traverse all the elements in $\tilde{\mathcal{C}}_{e'}$  and search for their nearest neighbors in $\tilde{\mathcal{C}}_{e}$; we return the smallest distance among all these nearest neighbor queries over $\tilde{\mathcal{C}}_{e'}$ as the solution of the BCP problem for the couple $(\tilde{\mathcal{C}}_{e}, \tilde{\mathcal{C}}_{e'})$. 


\begin{lemma} 
\label{lem-step2}
     The time complexity of Step (2) 
      is $O(nz \log \frac{\epsilon}{\delta} \cdot t_\mathtt{dis})$,  where $\delta$ is the smallest pairwise distance of $X_{\mathtt{in}}$.
\end{lemma}

\begin{proof}

    First, we need to build a cover tree for each $\tilde{\mathcal{C}}_e$, $e \in E$. According to Claim \ref{cover_insert}, we know that it takes $O(|\tilde{\mathcal{C}}_e| \log \Phi_{e} \cdot t_\mathtt{dis})$ time to build a cover tree on $\tilde{\mathcal{C}}_e$, where $\Phi_e$ is the aspect ratio of $\tilde{\mathcal{C}}_e$. Obviously, $\Phi_e \le \frac{\epsilon}{2\delta}$.  Therefore, the total time complexity  to build the cover trees for all $\tilde{\mathcal{C}}_e$s is:
    \begin{equation}
        \begin{aligned}
            \sum_{e\in E}  O(|\tilde{\mathcal{C}}_e|\log \Phi_e t_\mathtt{dis})= O( n \log \frac{\epsilon}{\delta} \cdot t_\mathtt{dis}).\label{for-time-step2-1}
        \end{aligned}
    \end{equation}
    Next, for each $\tilde{\mathcal{C}}_e$ and its each neighbor  $e'\in \mathcal{A}_e$, we solve the BCP problem of 
     $\tilde{\mathcal{C}}_e$ and $\tilde{\mathcal{C}}_{e'}$.  
    According to Claim \ref{cover_insert}, the overall time complexity of this procedure is 
    \begin{equation}
        \begin{aligned}
            &\sum_{e\in E} \sum_{e'\in \mathcal{A}_e} |\tilde{\mathcal{C}}_e| O(\log \Phi_{e'} \cdot t_\mathtt{dis})\\
            &= \sum_{e\in E} |\tilde{\mathcal{C}}_e| |\mathcal{A}_e| O(\log \frac{\epsilon}{\delta} \cdot t_\mathtt{dis}) \\
            &= O( nz \log \frac{\epsilon}{\delta} \cdot t_\mathtt{dis}).\label{for-time-step2-2}        \end{aligned}
    \end{equation}
    It is easy to see    the complexity (\ref{for-time-step2-2}) dominates the complexity (\ref{for-time-step2-1}). So we obtain the complexity of Step (2) that is $O( nz \log \frac{\epsilon}{\delta} \cdot t_\mathtt{dis})$. 
\end{proof}

\textbf{Step (3): identify the border points and outliers.} We can apply the similar idea used in Step (1). 
For a non-core point $p$, just find its nearest core point $c$ in  $\cup_{e\in\mathcal{A}_p}\tilde{C}_e$.  
If $\mathtt{dis}(p,c)\le \epsilon$, 
we put $p$ to the same cluster as $c$ and label it as a border point; otherwise, label $p$ as an outlier.
\begin{lemma}
\label{lem-step3}
The time complexity of Step (3)   is  $O( n z t_\mathtt{dis})$.
\end{lemma}

\begin{proof}
Similar to the proof of lemma \ref{lem-step1}, we define $E_1=\{e\in E\mid |\mathcal{C}_e|\ge MinPts\}$ and $E_2=E\setminus E_1$. 
  According to the definitions of border point and outlier, we know that 
  a non-core point cannot  locate inside some $\mathcal{C}_e$ with $e \in E_1$. 
  Therefore, the time complexity of Step (3) is
  \begin{equation}
     \sum_{e\in E_2}\sum_{{e'}\in \mathcal{A}_e}O(|\mathcal{C}_e||\mathcal{C}_{e'}|\cdot t_\mathtt{dis}),
  \end{equation}
  which is as same as the second term of equation (\ref{step2_com}).
  So by using the similar idea of Lemma~\ref{lem-step1}, we know  the time complexity of step (3) is also $O(nzt_\mathtt{dis})$.
\end{proof}


\textbf{The overall time complexity.} From the above Lemma~\ref{lem-step1}, Lemma~\ref{lem-step2}, and Lemma~\ref{lem-step3}, together with the complexity of Algorithm~\ref{exp_kcenter}, we know 
the overall time complexity for solving the DBSCAN problem under Assumption~\ref{ass-doubling2} is
$O\bigg{(}n\big{(}(\frac{\Delta}{\epsilon})^D + z\log \frac{\epsilon}{\delta}\big{)}t_\mathtt{dis}\bigg{)}$ which is linear in the input size $n$. Also note that usually $z\ll n$.  For example, in most of the test instances in our experiments and the previous articles ({{\em e.g.,}~\cite{gan2015dbscan}\cite{jang2019dbscan++}\cite{schubert2017dbscan}}), $z$ is less than $1\%$ $n$. 

\begin{remark} [\textbf{the setting for $\bar{r}$}]
 \label{rem-r}
 In our proposed method, we set $\bar{r} = \epsilon/2$ for simplicity when running Algorithm~\ref{exp_kcenter}. Actually, we can prove that any $\bar{r} \leq  \epsilon/2$ works for our algorithm. We also would like to emphasize that the pre-processing step ({\em i.e.,} Algorithm~\ref{exp_kcenter}) only needs to be performed once, if we have an estimated lower bound $\epsilon_0$ for $\epsilon$ and set $\bar{r} = \epsilon_0/2$. { When we magnify the parameter $\epsilon$ or adjust the parameter $MinPts$}, we only need to re-compute the $\mathcal{A}_p$ sets and run the above Step (1)-(3)   without running Algorithm~\ref{exp_kcenter} again. 
 So our method is quite efficient for parameter tuning, which is a significant advantage for practical implementation. { We also discuss this advantage in Section~\ref{prop_exp}.}

 \end{remark}

\subsection{When Outliers Also Have Low Doubling Dimension}
\label{covertree_dbscan}
In this section, we consider the scenario that the whole input data $X$ has a low doubling dimension $D$, which is a more stringent version of Assumption~\ref{ass-doubling2}. We show that the algorithm of the exact   DBSCAN  in Section~\ref{exact_dbscan-1}   can be further simplified.

\textbf{Our main idea.} With a slightly abuse of notations, we still use $\Delta$ and $\delta$ to denote the maximum and minimum pairwise distances of $X$, respectively. So the aspect ratio of $X$ is $\Phi=\Delta/\delta$. The main idea is to build a cover tree $T$ for the whole input data $X$ since we assume it has a low doubling dimension $D$.
Then, we can directly obtain the ball center set $E$ without running Algorithm~\ref{exp_kcenter}. Let $i_0=\lfloor \log \epsilon/2 \rfloor$. The set of nodes at the $i_0$-th level of $T$, which is denoted as $T_{i_0}$, is an $\bar{r}$-net with {$\bar{r} = 2^{\lfloor \log \epsilon/2 \rfloor}$ for $X$}. So we can simply let $E=T_{i_0}$. For each $p\in X$, we also define the neighbor ball center set $\mathcal{A}_p$ as (\ref{for-ap}). A key point here is that we can prove a much lower bound for $|\mathcal{A}_p|$.



{
\begin{lemma}
    If $(X,\mathtt{dis})$ is a metric space with constant doubling dimension $D$, then for any $p\in X$, we have $|\mathcal{A}_p| = O(1)$.
    \label{ap_const}
\end{lemma}
\begin{proof}
    Since $E$ is an $\bar{r}-net$, then $\min_{e_1, e_2 \in E} \mathtt{dis}(e_1,e_2) \ge \bar{r}$. 
For any $p_1, p_2 \in \mathcal{A}_p$, we have
  $\mathtt{dis}(p_1, p_2) \le \mathtt{dis}(p_1, p) + \mathtt{dis}(p, p_2)
      \le 4 \bar{r} + 2\epsilon$, where 
  the first inequality comes from the triangle inequality and the second inequality comes from the definition of $\mathcal{A}_p$. 
  According to Proposition~\ref{Ap_lemma} and using the same manner of the equality (\ref{alpha_forlula}), we have $|\mathcal{A}_p | \le 2^{D\lceil \log \alpha \rceil}$,
  where $\alpha \le \frac{4 \bar{r} + 2\epsilon}{r}$.
    So we have 
  $|\mathcal{A}_p| \le 2^{\lceil \log \frac{4 \bar{r} + 2\epsilon}{\bar{r}} \rceil D}
  = O\big{(}(\frac{\epsilon}{\bar{r}})^D\big{)}$. 
  By substituting $\bar{r} = 2^{\lfloor \log \epsilon/2 \rfloor} \ge 2^{ \log \epsilon/2 - 1} = \epsilon / 4$, we   immediately have $|\mathcal{A}_p| = O(4^D) = O(1)$.
\end{proof}
}

Then we can prove the time complexity of the exact DBSCAN algorithm with the new $O(1)$ bound of $|\mathcal{A}_p|$. In particular, we replace the previous $O(z)$ bound by $O(1)$ in Lemma~\ref{lem-step1}, Lemma~\ref{lem-step2}, and Lemma~\ref{lem-step3}, respectively. Finally, we obtain the following theorem.



\begin{theorem}
    If the metric space $(X,\mathtt{dis})$ has constant doubling dimension, our exact DBSCAN algorithm can be completed in $O(n \log \Phi $ $t_{\mathtt{dis}})$ time.
\end{theorem}

{
\color{black}

}



{\color{black}
\subsection{Remark on The Comparison with \texorpdfstring{\cite{ding2021metric}}{}}
\label{remark33}    
As mentioned in Section~\ref{sec-intro}, \citet{ding2021metric} proposed a randomized k-center clustering based DBSCAN algorithm. 
Here, we elaborate on the differences between our proposed method and their algorithm. 


The first major difference comes from the way for pre-processing input data. In \cite{ding2021metric}, they use a k-center with outliers algorithm to partition the data into $k$ balls. But their algorithm needs to set the parameter ``$\tilde{z}$'' as an estimated upper bound  for the number of outliers, and the termination condition has to be manually set in practice. 
Thus a drawback of this method is that any improper parameter setting ({\em e.g.}, underestimate the number of outliers) may lead to high computational complexity or even returning an incorrect DBSCAN result. 
In contrast, our pre-processing method, the Radius-guided Gonzalez algorithm proposed in Section~\ref{gonzalez}, does not need to set those parameters and only 
constructs an $\epsilon/2$ -net in the data space. 
Moreover, our method is a deterministic approach while the method of \cite{ding2021metric} is based on a ``randomized'' k-center with outliers algorithm that always has a failure probability.

Secondly, we have fundamentally different ways for constructing the DBSCAN clusters. Recall the three steps of DBSCAN as described in Section~\ref{exact_dbscan-1}. 
Comparing with the original DBSCAN algorithm~\cite{ester1996density}, the only improvement proposed by \cite{ding2021metric} is  using the aforementioned randomized k-center with outliers algorithm to reduce the search range when performing brute force search for $\epsilon$-neighborhoods of each point. This is a heuristic improvement for Step (1), and its time complexity remains $O(n^2)$ in the worst case. 
In contrast, our proposed algorithm attempts to establish an $\epsilon/2$-net data structure and convert the DBSCAN problem to a connectivity problem among the $\epsilon/2$-net spheres. Specifically, in Step (1), we divide the $\epsilon/2$-net spheres into two categories: dense and sparse spheres {\color{black}(corresponding to $E_1$ and $E_2$ in the proof of Lemma \ref{lem-step1})}. Each point in dense spheres can be immediately marked as a core point, and then the complexity of Step (1) is reduced to be linear. In Step (2), instead of considering the connection between individual core points, we consider the connection between $\epsilon/2$-net spheres. Together with the cover tree technique, the time complexity of Step (2) is also improved to be linear. Consequently, our overall time complexity is linear in $n$.


Furthermore, another advantage of our techniques is that they can be extended to solve the approximate and streaming DBSCAN problems, as shown in the following Section~\ref{approx_dbscan}.

}

\section{\texorpdfstring{$\rho$}{}-approximate Metric DBSCAN}
\label{approx_dbscan}
In this section, we present a linear time $\rho$-approximate DBSCAN algorithm that can be applied to arbitrary metric space. 
The original $\rho$-approximate DBSCAN algorithm in \citep{gan2015dbscan} is based on the grid method in Euclidean space $\mathbb{R}^d$. However, their method cannot be directly extended to general metric space; moreover, 
it takes an $O(n \cdot (1/\rho)^{d-1})$ running time that is unaffordable in practice for large $d$. 

Comparing with our exact DBSCAN algorithm in Section~\ref{exact_dbscan-1}, though both of the two algorithms have the time complexities linear in $n$, the approximate algorithm compresses the data into a small-size summary and thus it can be easily implemented for streaming data.  {Also in practice (as shown in our experiments), the approximate algorithm usually is faster than the exact algorithm. 
The main reason is that for exact DBSCAN, we need to solve a number of BCP problems when merging the core points (Step (2) in Section~\ref{exact_dbscan-1}). Although this step can be accelerated by using the cover tree, it is still relatively time-consuming in practice. In our approximate DBSCAN algorithm we can avoid this step by using the ``summary'' idea as shown below. 
}

\subsection{Linear Time Algorithm via Core Points Summary}
\label{sec-summary}

\textbf{Our main idea.} As shown in Section~\ref{exact_dbscan}, we can apply our proposed radius-guided Gonzalez's algorithm to efficiently reduce the complexity for labeling core points, border points, and outliers.
To merge the core points to form the DBSCAN clusters, we consider to build a ``summary'' $\mathcal{S}_*$ for the set of core points $\cup_{e\in E}\tilde{C}_e$, which is different to the cover tree idea. The summary set should satisfy two conditions: (1) $|\mathcal{S}_*|\ll |\cup_{e\in E}\tilde{C}_e|$ and (2) we can correctly generate the $\rho$-approximate clusters (as Definition~\ref{def-rhoapprox}) through $\mathcal{S}_*$. Our idea for constructing $\mathcal{S}_*$ is as follows. 
We  check each point $e\in E$ and its corresponding set $\mathcal{C}_e$ (obtained in Algorithm~\ref{exp_kcenter} with an appropriate $\bar{r}$). If $e$ is a core point, we add it to $\mathcal{S}_*$ and ignore all the other points of $\mathcal{C}_e$; otherwise, we add all the  core points  $\tilde{\mathcal{C}}_e$ to $\mathcal{S}_*$. The details are shown in Algorithm~\ref{fa_dbscan}. 

Though the above construction method for $\mathcal{S}_*$ is simple, it is challenging to prove that it satisfies the two conditions. Below, we first show its correctness, that is, we can generate the $\rho$-approximate clusters   through $\mathcal{S}_*$. 
{
Here we slightly modify the previous definition of the \textbf{neighbor ball center set} $\mathcal{A}_p$ in (\ref{for-ap}). Recall that for any $p \in X$, we define 
$\mathcal{A}_p = \{e | e \in E, \mathtt{dis}(e,c_p) \le 2 \bar{r} + \epsilon\}$. Now we change it to 
\begin{eqnarray}
 \mathcal{A}_p = \{e | e \in E, \mathtt{dis}(e,c_p) \le 4 \bar{r} + \epsilon\}. \label{for-ap-new}
 \end{eqnarray}
 Since we enlarge the set $\mathcal{A}_p$ in this modification, the result of Lemma~\ref{neighbor_lemma} is still true. 
 Also, this modification does not change the asymptotic upper bound $O\big{(}(\frac{\epsilon}{\bar{r}})^D + z\big{)}$ of $|\mathcal{A}_p|$  in Lemma~\ref{Ap_num}. We set $\bar{r} = \frac{\rho \epsilon}{2}$ in  Algorithm~\ref{fa_dbscan}, and then  the equation (\ref{approx_ap}) is replaced by:
\begin{equation}
\label{approx_ap-new}
    |\mathcal{A}_p| = O((\frac{1}{\rho})^D + z).
\end{equation}

}


\vspace{-10pt}
\begin{algorithm}
  \SetAlgoLined
  \caption{\sc{Metric $\rho$-approximate DBSCAN via Core points Summary}}
  \label{fa_dbscan}
  \KwIn{$\epsilon$, $MinPts$, $\rho$, $X$}
  
  Run Algorithm \ref{exp_kcenter} with  $\bar{r} = \frac{\rho \epsilon}{2}$ and $X$, let $E$ be the output set.  {We can also identify the core points of $E$ after running Algorithm \ref{exp_kcenter} (see our explanation in the proof of Lemma~\ref{step1}).}
  Initialize $\mathcal{S}_*\leftarrow \emptyset$.
  \tcp*[r]{the summary for core points}

\For{{\rm \textbf{each}} $e\in E$}{
 
    \eIf{$e$ is a core point}
    {
      Add $e$ to $\mathcal{S}_*$.
    }
    {
       Add all the core points in $\mathcal{C}_e$ to $\mathcal{S}_*$. 
    }
  }
  
 \textbf{Merge inside $\mathcal{S}_*$:} 
 for any pair of $c_1, c_2\in \mathcal{S}_*$, label them with the same cluster ID if   
  $\mathtt{dis}(c_1,c_2) \le (1+\rho)\epsilon$. 
  \label{merge_core_points}
  
  \textbf{Label other points:} 
    \For{$p\in X\setminus\mathcal{S}_*$}{
    \label{label_other_points}
        \eIf{$c_p \in \mathcal{S}_*$}
        {
            Label $p$ with the same cluster ID of $c_p$.
        }
        {
            \eIf{there exists some $s\in \mathcal{S}_*$ such that $\mathtt{dis}(p,s)\le (\frac{\rho}{2} + 1)\epsilon$}
            {
            Label $p$ with the same cluster ID of  $s$.
            }
            {
            Label $p$ as an outlier.
            }
        }
    }\label{end_label_other_points}
  

\end{algorithm}
\vspace{-8pt}







\begin{theorem}[Correctness]
  \label{dbscan_correctness}
 Algorithm \ref{fa_dbscan} can correctly return a $\rho$-approximate DBSCAN solution.
\end{theorem}
%
%
%
  
%
%
\begin{proof}
We prove the correctness following Definition~\ref{def-rhoapprox}. 
The \textbf{connectivity} is  easy to verify, since any two points are connected only when their distance is no larger than $(1+\rho)\epsilon$ in the algorithm. 
So we only need to prove the \textbf{maximality} of the obtained clusters. Namely,  $\forall p, q$, if $q$ is a core point and $p$ is density-reachable from $q$, then $p$ should have the same cluster ID with $q$.

To prove the \textbf{maximality}, we need to guarantee 
the uniqueness of the cluster ID for each core point $q$. If $q\in \mathcal{S}_*$, the uniqueness is satisfied when we merge inside $\mathcal{S}_*$; else, $q$ has the same cluster ID with $c_q$. So the uniqueness is guaranteed for all the core points, which implies that every core point belongs to exactly one cluster. Then we consider any two points $p, q\in X$ with $\mathtt{dis}(p, q)\leq \epsilon$. Let $q$ be a core point and $s_q$ be the ``representative'' of $q$ in $\mathcal{S}_*$ (according to the construction of $\mathcal{S}_*$, we let $s_q=c_q$ if $c_q$ is a core point; otherwise, $s_q=q$). 

\textbf{Case 1:} If $p$ is also a core point (its representative in $\mathcal{S}_*$ is denoted by $s_p$), then from the triangle inequality we know 
\begin{equation}
    \mathtt{dis}(s_p, s_q) \le \mathtt{dis}(s_p, p) + \mathtt{dis}(p, q) + \mathtt{dis}(q, s_q) \le (1 + \rho)\epsilon. \label{for-dbscan_correctness-1}
\end{equation}
It implies that $s_p$ and $s_q$ should have the same cluster ID; together with the ID uniqueness, we know $p$ and $q$ should have the same ID as well. 

\textbf{Case 2:} If $p$ is not a core point, actually $p$ should be a border point for this case.   
If $c_p \in \mathcal{S}_*$
, from the similar manner of (\ref{for-dbscan_correctness-1}) we know $\mathtt{dis}(c_p, s_q) \le (1+\rho) \epsilon$.  
So $c_p$ and $s_q$ should be merged when we merge inside $\mathcal{S}_*$.  Hence, $p$ and $q$ have the same cluster ID. 
If $c_p \notin \mathcal{S}_*$,  we have
    $\mathtt{dis}(p,s_q) \le \mathtt{dis}(p, q) + \mathtt{dis}(q, s_q) \le (\frac{\rho}{2} + 1)\epsilon$.
So  we label $p$ with the same cluster ID with $s_q$.
Also, $s_q$ and $q$ have same cluster ID, and thus   $p$ and $q$ have the same cluster ID.

Overall, we know that the \textbf{maximality} is guaranteed from the above two cases. 
So algorithm \ref{fa_dbscan} returns a qualified  $\rho$-approximate DBSCAN solution.
\end{proof}
Now, we analyze the complexity of Algorithm \ref{fa_dbscan}, { where the key is to prove the upper bounds for the complexity of constructing the summary  $\mathcal{S}_*$} and the induced query complexity in the labeling procedure. The  idea is similar with the proof of Lemma~\ref{lem-step1}, which considers the amortized complexity over the whole input instance.

\begin{theorem}[Complexity]
  \label{dbscan_com}
  Let $\rho \le 2$. 
  Algorithm~\ref{fa_dbscan} runs in  $O\Big{(}n\big{(}(\frac{\Delta}{\rho \epsilon})^D + z\big{)} \cdot t_\mathtt{dis}\Big{)}$ time.
\end{theorem}
 Usually $\rho$ is a small pre-specified number, since it measures the approximation degree and we do not want the result to be far away from the exact DBSCAN. So we  suppose $\rho\leq 2$ here (actually our analysis also works for the case $\rho>2$ with slight modification). 
To prove Theorem~\ref{dbscan_com}, we need to introduce some key lemmas first.  

\begin{lemma}
  \label{ce_num}
  For each $e\in E$, $|\mathcal{C}_e \cap \mathcal{S}_*| \le MinPts$. 
\end{lemma}
\begin{proof}
  If $e$ is a core point, then $\mathcal{C}_e \cap \mathcal{S}_* = \{e\}$ and obviously the size is no more than $MinPts$.

 If $e$ is not a core point, it implies $|\mathbb{B}(e, \epsilon)| \le MinPts$.
  Since we let $\rho \le 2$, then  $\frac{\rho \epsilon}{2} \le \epsilon$ and $\mathcal{C}_e \subseteq \mathbb{B}(e, \epsilon)$.
  Therefore, 
  \begin{equation}
    \label{cee_num}
    |\mathcal{C}_e \cap \mathcal{S}_*| \le |\mathcal{C}_e| \le |\mathbb{B}(e, \epsilon)| \le MinPts.
  \end{equation}
\end{proof}
\begin{lemma}
  \label{s_num}
  The size of $\mathcal{S}_*$ is $O\big{(}(\frac{\Delta}{\rho \epsilon})^D + z\big{)}$.
\end{lemma}
\begin{proof}
  The size of $\mathcal{S}_*$ is

   \begin{equation}
    \begin{aligned}
      |\mathcal{S}_*| &= \sum_{e \text{ is core point}} 1 + \sum_{e \text{ is not core point}} O(|\mathcal{C}_e|). 
    \end{aligned}
  \end{equation}
According to the  inequality (\ref{cee_num}), we have $|\mathcal{C}_e| \le MinPts$ { if $e$ is not a core point}, and thus 
  \begin{equation}
    \label{s_num2}
    \begin{aligned}
     |\mathcal{S}_*| &= O(|E|) + \sum_{e \text{ is not core point}} O(MinPts) \\
      &= O(MinPts \cdot |E|) \\
      &= O(|E|) 
      = O\big{(}(\frac{\Delta}{\rho \epsilon})^D + z\big{)},
    \end{aligned}
   \end{equation}
   {where the last equality follows from Lemma~\ref{kcenter_iteration}.}
\end{proof}

\begin{lemma}
  \label{step1}
  The construction of $\mathcal{S}_*$ can be completed in $O(((\frac{1}{\rho})^D + z)nt_\mathtt{dis})$ time.
\end{lemma}
  \begin{proof}
    For each $e\in E$, we can determine whether $e$ is core point or not when running Algorithm \ref{exp_kcenter} without introducing additional time complexity, since we have the distance between any pair of $e\in E$ and $p\in X$. 
    Hence, in the step of constructing $\mathcal{S}_*$, if $e$ is a core point, it takes $O(1)$ time to add $e$ to $\mathcal{S}_*$.

    If $e$ is not core point, we need to compute $|\mathbb{B}(p, \epsilon)|$ for every $p\in \mathcal{C}_{e}$.
   Based on Lemma~\ref{neighbor_lemma}, we only need to consider the set  $\cup_{{e'}\in \mathcal{A}_e} \mathcal{C}_{e'}$.
    In this case, determining whether $p$ is a core point takes $O(\sum_{e'\in \mathcal{A}_e} |\mathcal{C}_{e'}| t_\mathtt{dis})$ time.
    
   So the time complexity of constructing $\mathcal{S}_*$ (denoted by $t_c$) is
   \begin{equation}
     \begin{aligned}
       t_c&=\sum_{e \text{ is core point}}O(1) + \sum_{e \text{ is not core point}} \sum_{p\in \mathcal{C}_e} \sum_{e'\in \mathcal{A}_e} |\mathcal{C}_{e'}| t_\mathtt{dis} 
     \end{aligned}
    \end{equation}
 { Similar with the idea for proving lemma~\ref{ce_num}, we know for a non-core point $e\in E$, $|\mathcal{C}_e| \le MinPts$.} As a consequence, 
    \begin{equation}
      \begin{aligned}
     t_c&= O(|E|) + O(MinPts) \cdot \sum_{e \text{ is not core point}} \sum_{e'\in \mathcal{A}_e} |\mathcal{C}_{e'}|t_\mathtt{dis}.
                 \end{aligned}
            \end{equation}
We then exchange the summation order in terms of $e$ and $e'$ {(as the equation (\ref{exchange}))} to have
            \begin{equation}
            \begin{aligned}
     \sum_{e \text{ is not core point}} \sum_{e'\in \mathcal{A}_e} |\mathcal{C}_{e'}|t_\mathtt{dis}&=  \sum_{e'\in E} \sum_{\substack{e \in \mathcal{A}_{e'} \\ e \text{ is not core point}}} |\mathcal{C}_{e'}| t_\mathtt{dis}.
                      \end{aligned}
            \end{equation}
Note that $|\mathcal{C}_{e'}| t_\mathtt{dis}$ is independent with $e$, so we have
            \begin{equation}
            \begin{aligned}
     t_c&=O(|E|) + O(MinPts) \cdot \sum_{e'} O(|\mathcal{A}_{e'}|) |\mathcal{C}_{e'}|t_\mathtt{dis}\\
     &=O(n) + O(MinPts \cdot ((\frac{1}{\rho})^D + z)) \sum_{e'} |\mathcal{C}_{e'}|t_\mathtt{dis}\\
     &=O(((\frac{1}{\rho})^D + z)nt_\mathtt{dis}),
                           \end{aligned}
            \end{equation}
            from the facts    $\sum_{e'} |\mathcal{C}_{e'}| = O(n)$ and {$|\mathcal{A}_{e'}| = O((\frac{1}{\rho})^D + z)$ in (\ref{approx_ap-new})}. 
  \end{proof}

   \begin{lemma}
    \label{step2}
   The step of merging inside $\mathcal{S}_*$ needs  $O\Big{(}\big{(}(\frac{\Delta}{\rho \epsilon})^D + z\big{)}((\frac{1}{\rho})^D + z)t_{\mathtt{dis}}\Big{)}$ time.
   \end{lemma}
   \begin{proof}
{  In Algorithm \ref{fa_dbscan}, we set $\bar{r} = \frac{\rho \epsilon}{2}$ for Algorithm \ref{exp_kcenter}. Hence, for any pair of $s_1, s_2\in \mathcal{S}_*$ whose pairwise  distance 
  $\mathtt{dis}(s_1,s_2) \le (1+\rho)\epsilon$, we have $\mathtt{dis}(c_{s_1},c_{s_2}) \le \mathtt{dis}(s_1,s_2) + \mathtt{dis}(c_{s_1},s_1)+\mathtt{dis}(c_{s_2},s_2) \le (1+2\rho)\epsilon$. According to (\ref{for-ap-new}), we have $s_1\in \cup_{e\in \mathcal{A}_{s_2}} \mathcal{C}_e $ and vise versa. 
 Therefore, when we merge inside $\mathcal{S}_*$, for every $s\in \mathcal{S}_*$, we only need to search inside $ \big(\cup_{e\in \mathcal{A}_s} \mathcal{C}_e\big) \cap \mathcal{S}_*$ for the points who have the distance to $s$ no larger than
   $ (1+\rho)\epsilon$. }
   Hence, the time complexity for merging insider $S_*$ is 
   \begin{equation}
    \begin{aligned}
      t_m&=\sum_{s\in \mathcal{S}_*} \sum_{e\in \mathcal{A}_s} |\mathcal{C}_e \cap \mathcal{S}_*| t_{\mathtt{dis}}.
    \end{aligned}
  \end{equation}

 According to Lemma \ref{ce_num}, $|\mathcal{C}_e \cap \mathcal{S}_*| \le MinPts$ and thus
  \begin{equation}
    \begin{aligned}
      t_m&=\sum_{s\in \mathcal{S}_*}\! \sum_{e\in \mathcal{A}_s}\!\! O(MinPts) t_{\mathtt{dis}}
      =\sum_{s\in \mathcal{S}_*}  O(|\mathcal{A}_s| \cdot MinPts) t_{\mathtt{dis}}.
      \end{aligned}
    \end{equation}
  {   Based on (\ref{approx_ap-new})} we have
    \begin{equation}
      \begin{aligned}
      t_m&=\sum_{s\in \mathcal{S}_*} O\big(((\frac{1}{\rho})^D + z) \cdot MinPts\big) t_{\mathtt{dis}}\\
      &=O(|\mathcal{S}_*|) \cdot O(((\frac{1}{\rho})^D + z) \cdot MinPts)t_{\mathtt{dis}}\\
      &=O\Big{(}\big{(}(\frac{\Delta}{\rho \epsilon})^D + z\big{)}((\frac{1}{\rho})^D + z)t_{\mathtt{dis}}\Big{)},
    \end{aligned}
   \end{equation}
  { where the last equality follows from Lemma~\ref{s_num}. }
  \end{proof}

  \begin{lemma}
  \label{label_others_app}
      The step of labeling other points (borders and outliers) needs $n((\frac{1}{\rho})^D + z)t_\mathtt{dis}$ time.
  \end{lemma}
  Lemma~\ref{label_others_app} can be simply obtained from the proof of Lemma \ref{lem-step3} by substituting $|\mathcal{A}_p| = O(z)$ with $|\mathcal{A}_p| = ((\frac{1}{\rho})^D + z)$.

  { 
  Recall the time complexity of Algorithm \ref{exp_kcenter} with  $\bar{r} = \frac{\rho \epsilon}{2}$ is $O\Big{(}n$ $\big{(}(\frac{\Delta}{\rho \epsilon})^D + z\big{)} t_{\mathtt{dis}}\Big{)}$. 
  Overall, through combining Lemma~\ref{step1} - Lemma~\ref{label_others_app}, we obtain the complexity of Theorem~\ref{dbscan_com}.}

\begin{remark}
\label{rem-rho}
{  
} 
 Similar with the discussion in Remark~\ref{rem-r}, we do not need to run Algorithm~\ref{exp_kcenter} repeatedly for  $\rho$-approximate DBSCAN  when tuning the parameters in practice.
\end{remark}






\subsection{Implementation for Streaming Data}
\label{stream_dbscan}
We further consider designing streaming algorithm for $\rho$-approximate DBSCAN 
with the space complexity being independent of the size of $X$. 
The major challenge for designing a streaming  algorithm is that  the number of core points can be as large as  $O(n)$, which may result in large memory usage. Fortunately, we show that our proposed core point summary technique in Section~\ref{sec-summary} can help us to neatly circumvent this issue. 

\textbf{Sketch of our algorithm.} We implement Algorithm~\ref{fa_dbscan} in a ``streaming'' fashion. Our algorithm contains three stages. \textbf{Stage~$1$}: 
{
we apply a streaming incremental technique instead of directly running the radius-guided Gonzalez's algorithm. {Our main idea} is trying to assign every point in the data stream into existing balls within radius $\bar{r}$; if fail, we just build a new ball centered at this point. 
}
In this stage, we aim to obtain the ball centers $E$.
At the same time, we construct part of the summary $\mathcal{S}_*$ by adding the core points in $E$ into it. {We also need to keep a set {$\mathcal{M}=\{p \in X\mid \text{$c_p$ is not core point}\}$. Intuitively, the purpose of keeping this set is to identify the potential core points in next stage.}
\textbf{Stage $2$}: we identify the other {necessary core points that should be added to $S_*$}. {Namely, we add $p\in \mathcal{M}$ to $\mathcal{S}_*$ if $p$ is a core point.
This ensures the completeness of  the constructed summary $S_*$. }} 
Then, we can perform the merge on $\mathcal{S}_*$ offline in our memory. \textbf{Stage~$3$}: we scan $X$ in the final pass to assign other points {(as line \ref{label_other_points}-\ref{end_label_other_points} in Algorithm~\ref{fa_dbscan})}. Overall, we have the following Theorem~\ref{streaming_com}.
\begin{algorithm}[ht]

  \SetAlgoLined 
  \caption{\sc{Streaming $\rho$ -Approximate DBSCAN}}
  \label{sa_dbscan}
  \KwIn {$\epsilon$, $MinPts$, $X$, $\rho$}
  Initialize $E=\emptyset$, $\mathcal{M}=\emptyset$ and $\mathcal{S}_* = \emptyset$. 
  
  \textbf{The first pass:} 
  Scan $X$ , \For{every $p\in X$}{

    \If{there is no $e\in E$ such that $\mathtt{dis}(p, e) \le \bar{r}=\frac{\rho \epsilon}{2}$}
    {
        Add $p$ to $E$.
    }
        
            
            
        

    \For{every $e\in E$}{

    \uIf {$\mathbb{B}(e,\epsilon)$ has detected at least $MinPts$ points
  } 
  {$\mathcal{S}_* \leftarrow \mathcal{S}_* \cup \{e\}$.}
  \ElseIf{$\mathtt{dis}(p, e) \le \bar{r}=\frac{\rho \epsilon}{2}$} {$\mathcal{M} \leftarrow \mathcal{M} \cup \{p\}$.}
    }

  }
  
  {
  \textbf{The second pass:} While scanning $X$, identify the core points from $\mathcal{M}$ and add them to $\mathcal{S}_*$.

  }
  Merge inside $\mathcal{S}_*$ offline {as line \ref{merge_core_points} of Algorithm~\ref{fa_dbscan}.}

  \textbf{The third pass:} {Label the border points and outliers as line \ref{label_other_points}-\ref{end_label_other_points} of Algorithm~\ref{fa_dbscan}.}
\end{algorithm}
\vspace{-5pt}

\begin{theorem}
  \label{streaming_com}
   Let $\rho \le 2$, our streaming algorithm takes an $O\big{(}(\frac{\Delta}{\rho \epsilon})^D + z\big{)}$ memory usage (which is independent of $n$) and 
    an $O\Big{(}n\big{(}(\frac{\Delta}{\rho \epsilon})^D + z\big{)}\cdot t_\mathtt{dis}\Big{)}$ overall running time. 
\end{theorem}

\textbf{Sketch of the proof.} By using the similar idea of Lemma \ref{kcenter_iteration}, we can conclude that $O(|E|)=O\big{(}(\frac{\Delta}{\rho \epsilon})^D + z\big{)}$. And from the process of building $\mathcal{M}$, we know that $|\mathcal{M}|=O(MinPts\cdot|E|)=O(|E|)$. The elements in the set $\mathcal{S}_*$ are selected from $E$ and $\mathcal{M}$, so $\mathcal{S}_*$ does not occupy any extra space. Hence the overall memory usage is $ O(|E|+|\mathcal{M}|)=O\big{(}(\frac{\Delta}{\rho \epsilon})^D + z\big{)}$. Further, we can show that Algorithm \ref{sa_dbscan} can be completed in $O(n|E|t_\mathtt{dis})$ time, {\em i.e.,} its time complexity is $O\Big{(}n\big{(}(\frac{\Delta}{\rho \epsilon})^D + z\big{)}\cdot t_\mathtt{dis}\Big{)}$.

\section{Experiments}
\label{sec-exp}

\begin{table}[!ht] \small
\caption{Datasets}
    \label{dataset}
    \centering
    \begin{tabular}{cccc}
    \toprule
        \textbf{Dataset} & \textbf{Dimension} & \textbf{n} \\ \cmidrule(r){1-3}
        Moons~\cite{sklearn} & 2 & 10,000 \\ 
        Cancer\citep{Dua:2019} & 32 & 569 \\ 
        Arrhythmia\citep{guvenir1997supervised}  & 262 & 452 \\
        Biodeg\citep{mansouri2013quantitative}  & 41  & 1,055 \\
        MNIST
        \tablefootnote{
To enhance the density of high dimensional data in the space, 
we uniformly sampled 1000 data points from the original dataset and then duplicated them 10 times  with adding  {random noise in the range of $[-5, 5]$ for each dimension}. We perform the same operation to process CIFAR 10, USPS HW, and Fashion MNIST.}
\cite{lecun1998gradient} 
& 784 & 10000 \\
Fashion MNIST \cite{fashion_MNIST} & 784 & 10,000 \\ 
USPS HW\cite{hull1994database} & 256 & 10,000 \\ 
CIFAR 10\citep{cifar10} & 3072 & 10,000 \\ 
{
\color{black}
DEEP1B~\cite{babenko2016efficient}} & 96 & 9,990,000 \\
{\color{black} GIST \cite{jegou2010product}} & 960 & 1.000.000 \\ 
\color{black} GloVe25~\cite{pennington2014GloVe} & 25 & 1,183,514 \\ 
\color{black} SIFT\cite{jegou2010product} & 128 & 1,000,000 \\ 
\color{black}PCAM\cite{veeling2018rotation} & 1024 & 2,493,440 \\
\color{black}Spotify\_Session\cite{brost2019music} & 21 & 2,072,002,577 \\
\color{black}LSUN\cite{yu15lsun} & 1024 & 2,943,300 \\
COLA \cite{COLA} & n/a &515 \\ 
AG News \cite{ag_news} & n/a &7,600 \\ 
MRPC \cite{MRPC} & n/a &1,725 \\ 
MNLI \cite{MNLI} & n/a &9,815 \\

        
\bottomrule 
    \end{tabular}
\end{table}

        


\begin{figure*}[ht]
    \centering
    \includegraphics[width = 0.9\textwidth]{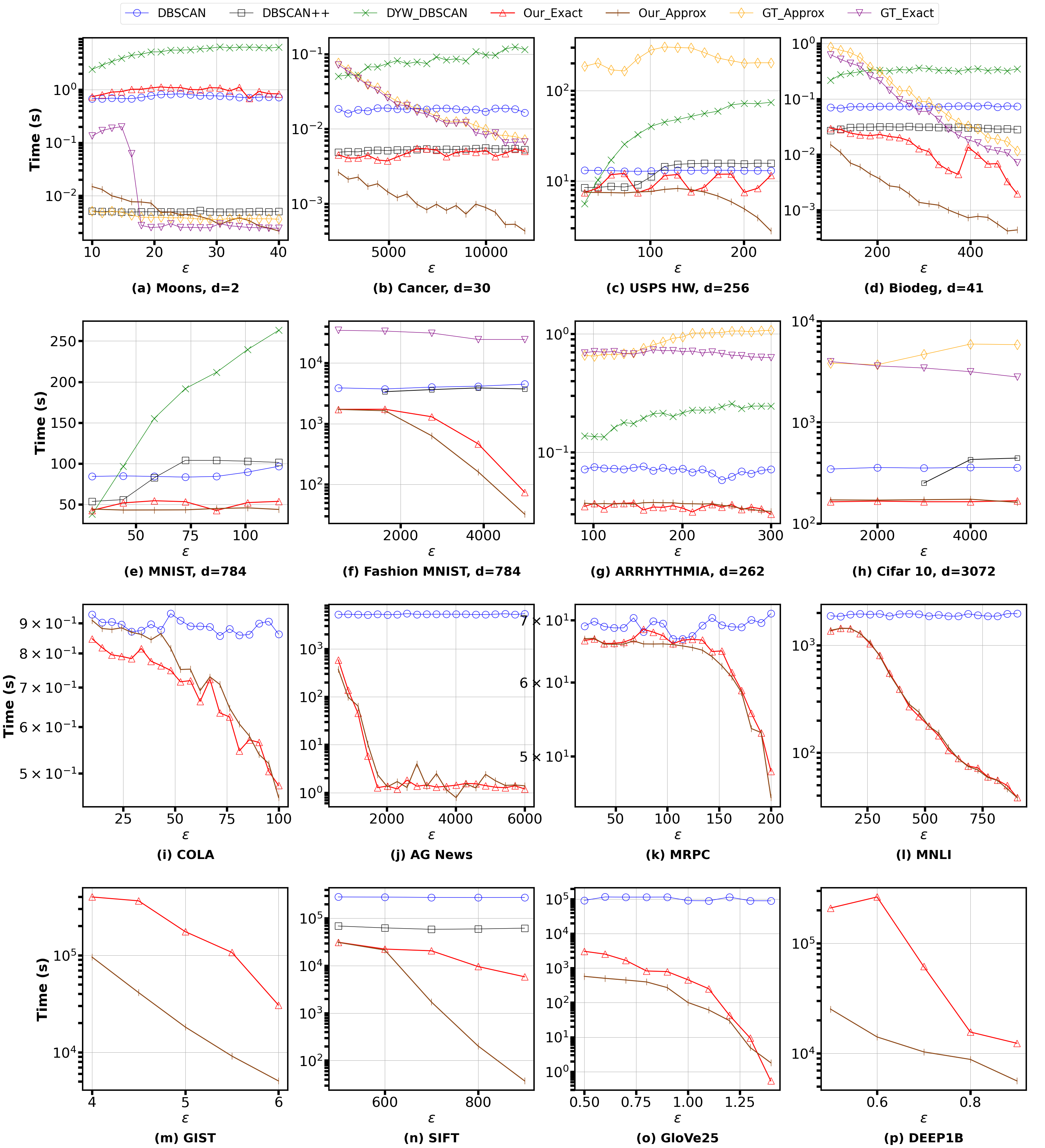}
    \caption{{Running time with varying $\epsilon$.} Some baseline algorithms are not tested in some figures,  because either they run too slowly ($>10^6$s) on the high-dimensional data, or they cannot run on the non-Euclidean data.\protect\footnotemark}
    \label{approx_time}
\end{figure*}

\subsection{Experimental environment and datasets}
\label{sec-data}

Our experiments were conducted on a server equipped with Intel(R) Xeon(R) Gold 6154 CPU @ 3.00GHz CPU and 512GB memory. We implement our algorithms in C++.

The datasets are shown in Table \ref{dataset} that includes both low dimensional and high dimensional datasets. 
In particular, to demonstrate the performance of our algorithm on general metric space, we also consider some non-Euclidean datasets, which include 4 real-world text datasets: {\textbf{AG News}~\cite{ag_news}, \textbf{COLA}~\cite{COLA}, \textbf{MNLI}~\cite{MNLI}, \textbf{MRPC}\cite{MRPC}. We use edit distance\cite{edit} to measure the distance between texts.} 
{\color{black} As shown in the table, we also take several million-scale datasets to evaluate the performance.}
In our experiments, each test instance is repeated 5 times and we report the average result.

\subsection{Comparison on Running time}
\label{sec-comptime}

{We compare the running time with several existing DBSCAN algorithms. In this section, we fix the parameters $MinPts=10$ and $\rho=0.5$, and adjust $\epsilon$ within a reasonable range for each dataset 
(more detailed discussion on $\rho$ is shown in Section \ref{sec-ari}). 
}

We compare the running time of our exact metric DBSCAN  (denoted by Our\_Exact) and approximate DBSCAN  (denoted by Our\_Approx) with the following algorithms: the original DBSCAN \citep{ester1996density}  (denoted by DBSCAN), DBSCAN++\cite{jang2019dbscan++},  { the metric DBSCAN  \citep{ding2021metric} (denoted by DYW\_DBSCAN)}, the exact and approximate DBSCAN \citep{gan2015dbscan} (denoted by GT\_Exact and GT\_Approx, respectively). {
Note that the clustering performance of the DBSCAN++ algorithm 
depends on the sampling ratio; we choose the ratio to be $30\%$ as suggested in their paper for achieving promising clustering results. 
}

We show the running time curves in Figure \ref{approx_time}. 
The datasets given in Table \ref{dataset} can be categorized into four classes corresponding to the four rows of Figure \ref{approx_time} from top to bottom: {4 low/medium-dimensional datasets}, 4 high-dimensional datasets , 4 text datasets, {\color{black}and 4 million-scale datasets. }
Overall, our exact and approximate algorithms can achieve much lower running time comparing with the baselines, especially for large high-dimensional and non-Euclidean  datasets. 
 {\color{black}For example, for the largest two datasets GIST and DEEP1B, only our algorithms can complete within $10^6$s (about 12 days) on our workstation;} for the dataset CIFAR 10, our algorithm runs in almost half the time of the original DBSCAN, and about 1/10 of GT\_Exact and GT\_Approx;  for the text dataset AG News, our algorithm takes less than 1\% of the time of the  original DBSCAN. Also,   our proposed approximate DBSCAN algorithm is   comparable or faster than our exact DBSCAN algorithm  in the experiments.

\subsection{{Discussion on clustering quality with \texorpdfstring{$\rho$}{}}} 
\label{sec-ari}


In this section, we discuss the impact of $\rho$ on the clustering quality. We run our approximation algorithm with a fixed $\epsilon$ and different values of $\rho$ and compare their clustering performance with the exact algorithm. 
{We choose four datasets (MNIST, USPS HW, Fashion MNIST, and CIFAR 10), and we take the widely used measures \emph{Adjusted Rand Index} (ARI)~\citep{hubert1985comparing} and \emph{Adjusted Mutual Information} (AMI)~\citep{vinh2009information}  to  evaluate their performances (higher values for ARI and AMI indicate better clustering quality). \footnotetext{In the version published at the SIGMOD '24 conference, we identified an error in the coding of the program, which resulted in incorrect displays of subfigures (i), (j), (k), and (l) in Figure 3. We have corrected this error in the current version of Figure 3. We apologize for this oversight.}

From \label{rho-start} Figure \ref{rect}  we can observe that when $\rho$ is set to 0.5, our approximate algorithm can achieve  similar clustering qualities with the exact algorithm for most instances. As an illustration, Figure~\ref{perform} (a)-(d) show two examples of the clustering results of exact DBSCAN and our approximate algorithm with $\rho=0.5$; we can see that the results are very close. \label{rho-end}

{\color{black}
\begin{remark}
\label{rem-ari}
    In a few cases of Figure \ref{rect}, we observe that the ARI or AMI score is slightly improved from $\rho=0.1$ to $\rho=0.5$. One possible explanation is that the goal of DBSCAN may not be exactly consistent
    with the best ARI and AMI values. The change of the clustering performance with the varying of $\rho$  can be complex. So when $\rho$ becomes 

\begin{figure}[]
    \centering
    \includegraphics[width=0.4\textwidth]{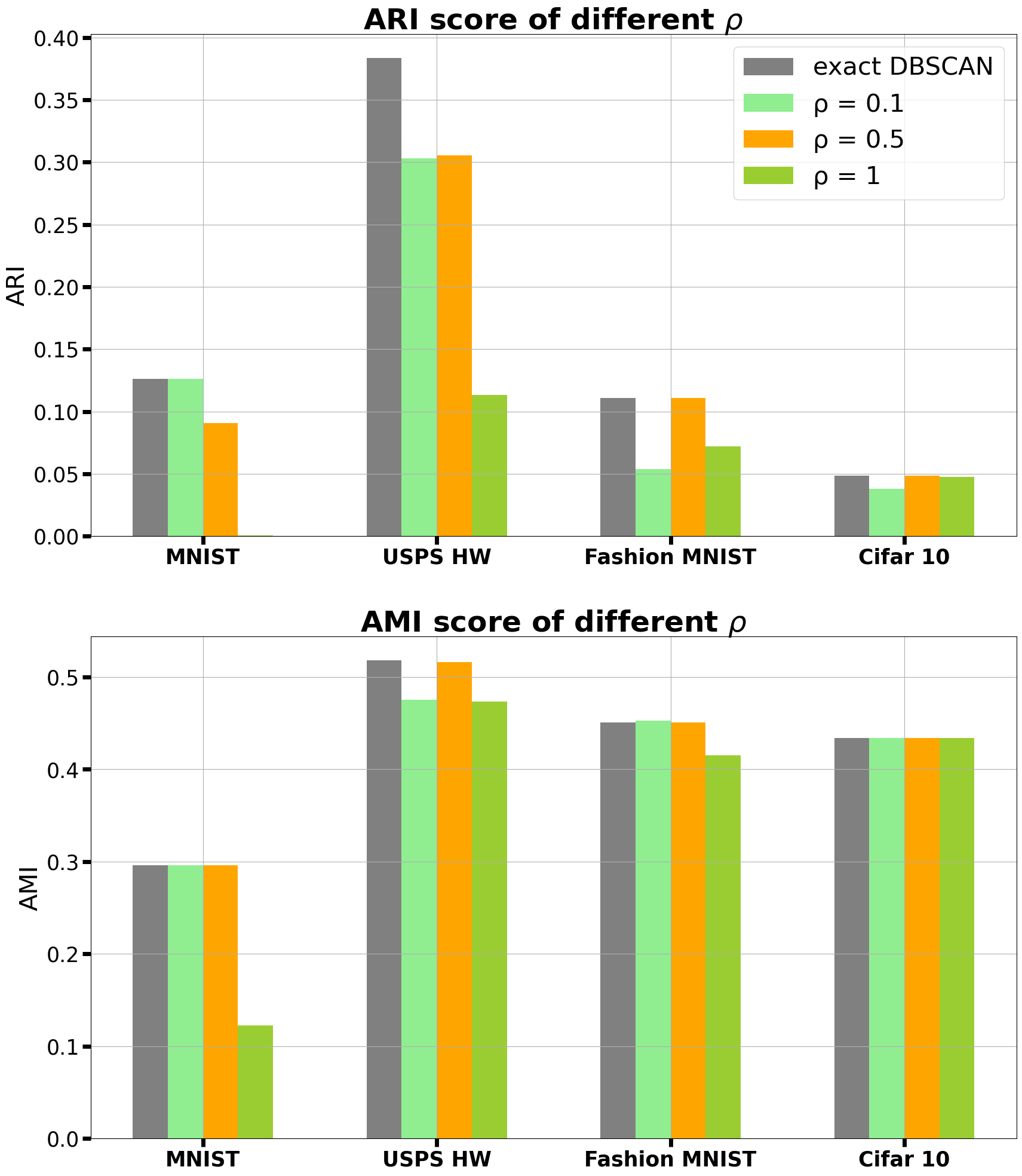}
    \caption{The ARI and AMI with fixed $\epsilon$ and different $\rho$.}
    \label{rect}
\end{figure}
\begin{figure}[]
    \centering
    \includegraphics[width = 0.4\textwidth]{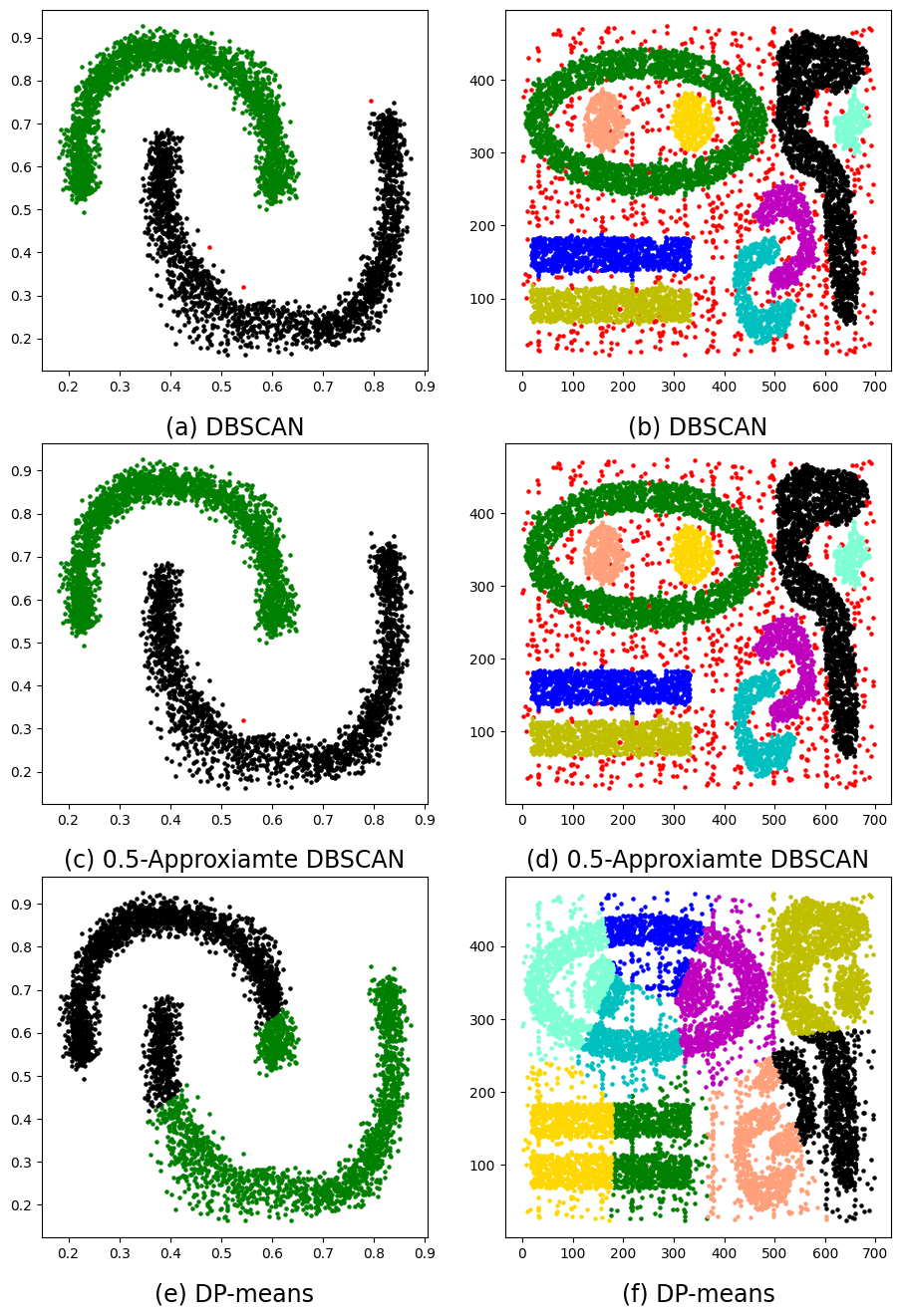}
    \caption{ Clustering results of exact DBSCAN, our approximate algorithm with $\rho=0.5$ and DP-means. The points with same color belong to the same cluster, and the red points are outliers.}
    \label{perform}
\end{figure}

\noindent smaller, though the clustering result is closer to the exact DBSCAN, it does not always have increased ARI or AMI. The experiments in \cite{jang2019dbscan++} also show that DBSCAN may not yield perfect ARI and AMI. 
\end{remark}
}

\vspace{-8pt}
{\color{black}
\subsection{Comparison with non-DBSCAN Algorithms}
\label{non-dbscan}
For completeness, we also consider the comparison with some non-DBSCAN clustering algorithms. Specifically, we consider the following baselines: {\color{black}``Dirichlet process (DP)-means''~\cite{kulis2011revisiting} (a non-parametric $k$-means clustering algorithm that does not require to input the number of clusters $k$)}, 
the streaming $k$-means algorithm ``BICO''~\cite{fichtenberger2013bico}, and two density clustering algorithms ``Density-peak''~\cite{rodriguez2014clustering} and ``Meanshift''~\cite{comaniciu2002mean} \footnote{We are aware of the recent improved algorithm Meanshift++~\cite{jang2021Meanshift++}, but Meanshift++ is only applicable to low-dimensional data, and thus is not suitable for most of our experimental datasets.}
.
For fairness in the comparison, we try our best to tune the parameters for each baseline. It is also worth noting that BICO  requires to manually specify the number of clusters $k$. In other words, the performance of BICO is likely to have downgrade in practical applications, if the number of clusters is not given. As for DP-means, the cluster penalty parameter $\lambda$ is set by taking the maximum distance of the k-center initialization, following the suggested setting in their original paper.

To evaluate their performances on high-dimensional dense datasets, we particularly construct the ``MNIST\_noisy'' and ``Fashion\_noisy'' datasets: we duplicate each point in the MNIST and Fashion MNIST datasets $10$ times and add a random perturbation in the range of $[-5, 5]$  to each dimension of each duplicated point, and then randomly generate $1\%$  noisy points within   $[0, 255]^d$ where $d$ is the dimensionality of the image space. 
From Table~\ref{k-means}, we can see that our exact and approximate DBSCAN outperform the baselines in terms of the ARI and AMI scores on most of the datasets. 

{\color{black}As for the runtime, we observe that BICO has similar speed with our algorithms (for example, the average difference on dataset MNIST is less than $20\%$), and the two baseline density clustering algorithms (Density-peak and Meanshift) are much slower (at least 5 times on average, and Density-peak even encounters memory overflow issue on some large datasets.).
The baseline DP-means, though usually has worse clustering performance than our algorithms as shown in Table~\ref{k-means} (two examples are also given in Figure~\ref{perform} (e)-(f)), is always the fastest one, which is about $12$ times faster on average than our algorithms and BICO.} We attribute this to the inherent algorithmic simplicity of DP-means than other testing algorithms (especially the density-based algorithms often need to conduct more complicated operations towards spatial neighborhood).


}

\begin{table}[ht] \small
\centering
\caption{{The runtime proportions of Algorithm~\ref{exp_kcenter} in our exact DBSCAN algorithm}}
\vspace{-0.1in}
\begin{tabular}{cccc} \\ \toprule
\textbf{Dataset}    & \begin{tabular}{c}\textbf{Radius-guided} \\ \textbf{Gonzalez (ms)}\end{tabular} & 
\begin{tabular}{c}
     \textbf{Total time} \\ \textbf{(ms)}
\end{tabular}
 & \textbf{Proportion} \\ \cmidrule(r){1-4}
Moons      & 4.16        & 4.44       & 94\%         \\
Cancer & 5.10      & 5.75         & 89\%         \\
USPS HW    &6090   &7994   &76\% \\

Biodeg     & 35.1     & 38.8         & 90\%          \\


MNIST   &40474  &43597  &93\% \\
Fashion MNIST     & 5113      & 5136          & 99\%         \\
Arrhythmia     & 35.1     & 38.8         & 90\%          \\
CIFAR 10      & 155966       & 156611        & 99\%         \\ 

COLA    &702   &779   &90\% \\
AG News    &73131   &114830   &64\% \\
MRPC & 67148    &67158 & 99\% \\

\bottomrule
\end{tabular}
\vspace{-10pt}
\label{kcenter-time}
\end{table}



  \begin{table*}[ht] \small
  \centering
  \captionsetup{justification=centering}
  \caption{\color{black} Comparison of ARI and AMI scores with the non-DBSCAN algorithms. \\The symbol "-" denotes values less than 0.01 and "*" denotes memory overflow (>500GB)}  
  \begin{tabular}{ccccccccccccccccc} \toprule
    \multirow{2}{*}{\textbf{Dataset}} & \multicolumn{2}{c}{
    \textbf{DBSCAN}} 
    & \multicolumn{2}{c}{
    \begin{tabular}{c}
         \textbf{0.5-approx}\\
         \textbf{DBSCAN} 
    \end{tabular}
    } &\multicolumn{2}{c}{\textbf{DP-means}} & \multicolumn{2}{c}{\textbf{BICO}}
    &\multicolumn{2}{c}{\textbf{Density-peak}} & \multicolumn{2}{c}{\textbf{Meanshift}}
\\
      & ARI     & AMI & ARI   & AMI  & ARI     & AMI  & ARI   & AMI & ARI     & AMI  & ARI   & AMI \\ \hline
Moons &\textbf{1.0}    & \textbf{1.0} &0.99 & 0.96  & 0.26 &0.30 &  0.19   & 0.19  & 0.51 &0.67& 0.39 &0.30\\
Cluto & 0.94 & 0.91 &\textbf{0.95} & \textbf{0.92} & 0.28 &0.40 & 0.59 & 0.50 & 0.41 &0.62&0.19 &0.33  \\
MNIST & 0.25 & 0.48 &0.28 & \textbf{0.51} & - & - & \textbf{0.31} & 0.41 & 0.01 &0.06& 0.02 &0.43 \\
MNIST\_Noisy & \textbf{0.28} & \textbf{0.57} &\textbf{0.28} & 0.51 & - & 0.05 & 0.17 & 0.22 & * &* & 0.02 & 0.49  \\
Fashion & 0.14 & 0.46 &\textbf{0.25} & \textbf{0.50} & - &- & 0.23 & 0.39 & 0.02 &0.06& 0.02 &0.43\\
Fashion\_Noisy & 0.15 & 0.52 &\textbf{0.24} & \textbf{0.52} &- &0.05 & 0.14 & 0.38 & * & * & 0.02 &0.49\\
PCAM & 0.02 & \textbf{0.13} &0.02 & \textbf{0.13} & \textbf{0.05} &0.04 & 0.01 & 0.02 & * &*& - &-\\
LSUN & 0.03 & \textbf{0.13} &0.02 & 0.12 & \textbf{0.04} &0.02 & 0.01 & 0.02 & * &*& - &-\\


\bottomrule
  \end{tabular}
  \label{k-means}
  \end{table*}

}







\subsection{Time Taken by the Radius-guided Gonzalez}
\label{prop_exp}


As discussed in Remark \ref{rem-r} and Remark~\ref{rem-rho}, we do not need to repeatedly run Algorithm~\ref{exp_kcenter} when tuning the parameters $\epsilon$ or $MinPts$. This property is quite useful in practical implementation because the time of Algorithm~\ref{exp_kcenter} often takes a large part of the whole DBSCAN procedure. To verify this, we illustrate the  runtime proportions taken by Algorithm~\ref{exp_kcenter}  in {our exact DBSCAN algorithm} in 
Table~\ref{kcenter-time}.

\begin{table*}[ht] \small
  \centering
  \caption{Results of the streaming algorithms, the symbol "-" denotes values less than 0.01}
  \begin{tabular}{ccccccccccc} \toprule
    \multirow{2}{*}{\textbf{Dataset}} & \multicolumn{2}{c}{
    \textbf{Our algorithm}} 
    & \multicolumn{2}{c}{\textbf{DBStream}} & \multicolumn{2}{c}{\textbf{D-Stream}} & \multicolumn{2}{c}{\textbf{evoStream}} & \multicolumn{2}{c}{\textbf{BICO}} \\
      & ARI     & AMI    & ARI   & AMI  & ARI     & AMI  & ARI   & AMI  & ARI   & AMI \\ \hline
Moons &\textbf{0.97}    & \textbf{0.90}  & 0.42  & 0.38  &  0.19   & 0.19  & 0.23  & 0.30  & 0.21  & 0.24 \\
Cancer & \textbf{0.70} & \textbf{0.58}  & 0.13 & 0.13  & 0.59 & 0.50  & 0.44 & 0.40 & 0.53 & 0.40 
\\
Arrhythmia                        & \textbf{0.23}        & \textbf{0.13}               & 0.16          & 0.13        & 0.01              & 0.01             & 0.18               & 0.12             & 0.05        & 0.10     \\
Biodeg &\textbf{0.09} &\textbf{0.06} &- &0.03 &0.01 &0.02&0.01 &0.05 &0.01 &0.04 \\
MNIST & 0.28    & \textbf{0.51}              & 0.04              & 0.12             & -              & -            & 0.01               & 0.07            & \textbf{0.31}        & 0.41   \\
CIFAR 10                              & 0.02                 & \textbf{0.43}      & 0.04          & 0.07                 & 0.01              & 0.02            & -               & 0.01             & \textbf{0.03}        & 0.13     \\
Fashion MNIST & \textbf{0.26} & \textbf{0.50} & - & - & - & - & - & - & 0.23 & 0.39
\\
USPS HW                             & 0.19                 & \textbf{0.53}               & -          & -                  & 0.01   & 0.10  & 0.22      & 0.40    & \textbf{0.37}        & 0.49\\
\color{black} PCAM & \textbf{0.02} & \textbf{0.03} & - & -  & - & - & - & -  &0.01 & 0.02  \\
\color{black} LSUN & - & \textbf{0.04} & - & -  & - & - & - & -  &\textbf{0.01} & -  \\
\color{black}Spotify\_Session 1\%  & \textbf{0.02}                 & \textbf{0.14}              & 0.02          & -            & -   & -     &-     & -      & 0.01        & -  \\ 
\color{black}Spotify\_Session 10\%  & \textbf{0.02}                 & \textbf{0.15}              & 0.02          & -              & -   & -   &-     & -    & 0.01        & - \\
\color{black}Spotify\_Session 50\%  & \textbf{0.02}                 & \textbf{0.15}             & 0.02          & -             & -   & -    &-     & -    & 0.01        & -  \\ 
\color{black}Spotify\_Session 100\%  & \textbf{0.02}                 & \textbf{0.09}             & 0.02          & -            & -   & -  &-     & -     & 0.01        & - \\
\bottomrule
  \end{tabular}

  \label{streaming_ari} 
  \end{table*}

From Table~\ref{kcenter-time} below, we can see that the runtime of Algorithm~\ref{exp_kcenter} takes more than 60\%  of the total DBSCAN procedure; it indicates that when we increase $\epsilon$ or adjust $MinPts$ during parameter tuning,  a large amount of of runtime can be saved without repeatedly running Algorithm~\ref{exp_kcenter}. 
{ In our approximate algorithm, this ratio is even higher,  \emph{e.g.}, in our approximate algorithm on the MNIST dataset, the Radius-Gonzalez procedure takes almost 98\% of the total time. We  leave the detailed results 
to our full version}.

\begin{figure*}[ht]
    \centering
    \includegraphics[width=0.9\textwidth]{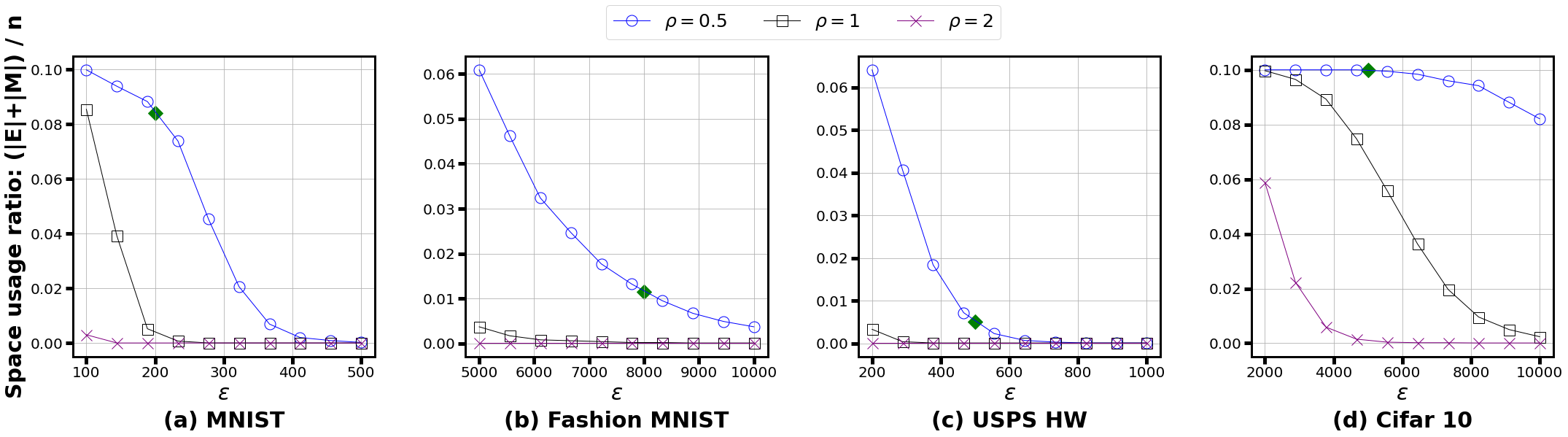}
    \vspace{-5pt}
    \caption{Memory usage of our streaming algorithm; the green diamond in each figure corresponds to the parameter that we use in Table~\ref{streaming_ari}\vspace{-10pt}}
    \label{streaming-space}
\end{figure*}

\subsection{Streaming DBSCAN Algorithm}
\label{sec-expstreaming}
We compare the clustering result of our streaming algorithm ({\em i.e.,} Algorithm~\ref{sa_dbscan}) with several popular streaming density clustering algorithms such as DBStream \citep{hahsler2016clustering}, D-Stream \citep{chen2007density}, evoStream \citep{carnein2018evoStream}, and BICO~\citep{fichtenberger2013bico}.
We set $\rho = 0.5$ for our algorithm.
For each baseline algorithm, we show its performance under the properly tuned parameters.
{\color{black} We also take some real stream data with large scales: the billion-scale dataset Spotify\_Session, and the two million-scale datasets PCAM and LSUN (please refer to Table~\ref{dataset} for details).
Additionally, we divide the Spotify\_session dataset into four datasets by date, namely the earliest 1\%, 10\%, 50\%, and the full set; since the recorded data in Spotify\_Session may have changing trend over time, we can view them as four different datasets.} We also use ARI and AMI  to evaluate their performances.  The results are shown in Table \ref{streaming_ari}.
{We can see that for most of the test instances, our streaming algorithm achieves better quality over  other baselines. 
We also consider   the memory usage of our streaming algorithm. 
Recall that in Algorithm \ref{sa_dbscan}, we need to store the set $E$ and $\mathcal{M}$ ($\mathcal{S}_* \subseteq E \cup \mathcal{M}$, so that we do not need extra memory for $\mathcal{S}_*$) in memory.
So we evaluate the memory usage by the ratio   $\frac{|E|+|\mathcal{M}|}{n}$. 
For each dataset, we conducted the experiments with different values of $\rho\in \{0.5$, $1$, $2\}$, and the results are shown in Figure \ref{streaming-space}. 
We can see that our algorithm can greatly save the memory usage, \emph{e.g.}, for the Fashion MNIST dataset, our algorithm only needs to store about 1\% of the data points (see the green diamond in the curve).


}

\section{Conclusion and Future work}
\label{sec-con}

In this paper, we study the metric DBSCAN problem and present its exact, approximate, and streaming algorithms.
We first study their quality guarantees in theory, and then conduct a set of experiments to compare  with other DBSCAN algorithms.
{\color{black}
{\color{black} Following this work, there are several opportunities to further improve our methods from both theoretical and practical perspectives.} For example,   is it possible to design a faster Gonzalez's algorithm under Assumption~\ref{ass-doubling2}? We are aware that a fast Gonzalez's algorithm was proposed by~\citet{DBLP:journals/siamcomp/Har-PeledM06} in doubling metric, but they assume that  the whole input data (including both inliers and outliers) has a low doubling dimension. As discussed in Section~\ref{non-dbscan}, we will also consider to reduce the overall time complexity by designing new algorithmic techniques, such as  
new $r$-net method with lower complexity. In particular, the improvement on efficiency may lead to real-world applications for database systems research ({\em e.g.,} preprocessing large-scale NLP database to support training machine learning model more efficiently).  
As for streaming DBSCAN, we may consider to reduce the number of passes for our streaming algorithm via developing some more succinct  data structure in memory; it also has certain practical significance to deal with other operations and issues, like data deletion and  drift.  
 

}

\section{Acknowledgement}
The research of this work was supported in part by the National Natural Science Foundation of China 62272432, the National Key Research and Development Program of China 2021YFA1000900, and the Natural Science Foundation of Anhui Province 2208085MF163.
We want to thank the anonymous reviewers for their helpful comments. We are grateful to Qingyuan Yang for his help on refining our code and for inspiring discussions.

\bibliographystyle{plainnat}
\bibliography{main}

\received{October 2023}
\received[revised]{January 2024}
\received[accepted]{February 2024}

\end{document}